\documentclass[a4paper,10pt]{article}
\usepackage[T1]{fontenc}
\usepackage{hyperref,amsmath,amssymb,amsthm,enumerate,xcolor,mathtools,cleveref,mleftright,orcidlink,authblk,feyn}
\usepackage[final,spacing]{microtype}
\usepackage[german,french,czech,russian,latin,british]{babel}
\usepackage{CJKutf8}
\usepackage[osf]{newpxtext}
\usepackage[vvarbb,slantedGreek]{newpxmath}
\usepackage[artemisia]{textgreek}
\usepackage[utf8]{inputenc}
\renewcommand\Omega\Omegaup
\renewcommand\gamma\gammaup
\renewcommand\delta\deltaup

\usepackage{tikz-cd}
\usetikzlibrary{decorations.pathmorphing}
\usetikzlibrary{babel}
\newtheorem{lemma}{Lemma}
\newtheorem{theorem}{Theorem}
\theoremstyle{definition}
\newtheorem{definition}{Definition}
\theoremstyle{remark}
\newtheorem{example}{Example}
\newtheorem{remark}{Remark}

\renewcommand{\fg}{\mathfrak{g}}
\newcommand{\fp}{\mathfrak{p}}
\newcommand{\ft}{\mathfrak{t}}
\newcommand{\fh}{\mathfrak{h}}
\newcommand{\fG}{\mathfrak{G}}
\newcommand{\fder}{\mathfrak{der}}
\newcommand{\ZZ}{\mathbb{Z}}
\newcommand{\Q}{\mathcal{Q}}
\newcommand{\CEc}{\operatorname{CE}_\bullet}
\newcommand{\CE}{\operatorname{CE}^\bullet}
\newcommand{\Coder}{\operatorname{Coder}}
\newcommand{\Sym}{\operatorname{Sym}}
\hypersetup {
    pdftitle = {Twisted Homotopy Algebras: Supersymmetric Twists, Spontaneous Symmetry Breaking, Anomalies and Localisation},
    pdfauthor = {Leron Borsten, David Simon Henrik Jonsson, Dimitri Kanakaris Decavel, and Hyungrok Kim},
    pdfsubject = {hep-th, math-ph},
    pdflang = {en-GB},
    pdfkeywords = {},
}

\title{Twisted Homotopy Algebras: Supersymmetric Twists, Spontaneous Symmetry Breaking, Anomalies and Localisation}
\author[a,b]{Leron Borsten\textsuperscript{\orcidlink{0000-0001-9008-7725}}}
\author[b,c]{Simon Jonsson\textsuperscript{\orcidlink{0009-0001-7155-8496}}}
\author[b]{Dimitri Kanakaris\textsuperscript{\orcidlink{0009-0001-7716-851X}}}
\author[b]{Hyungrok Kim~(\begin{CJK*}{UTF8}{bsmi}金炯錄\end{CJK*})\textsuperscript{\orcidlink{0000-0001-7909-4510}}}

\affil[a]{Blackett Laboratory, Imperial College London, London \textsc{sw7 2az}, United Kingdom}
\affil[b]{Department of Physics, Astronomy and Mathematics, University of Hertfordshire, Hatfield, Hertfordshire \textsc{al10 9ab}, United Kingdom}
\affil[c]{Aros Kapital, Vestagatan 6, 416 64 Göteborg, Sweden }
\newcommand\email[1]{\texttt{\href{mailto:#1}{#1}}}
\affil[ ]{\email{l.borsten@herts.ac.uk}, \email{simon.jonsson010@gmail.com}, \email{d.kanakaris-decavel@herts.ac.uk}, \email{h.kim2@herts.ac.uk}}

\begin{document}
\maketitle
\begin{abstract}
Twisting and classical background fields are two foundational techniques in supersymmetric quantum field theory, central to developments ranging from the Higgs mechanism to topological twisting and supersymmetric localisation. While traditionally treated as distinct procedures, they appear on an equal footing in the homotopy-algebraic approach to quantum field theory.
In this work, we formalise this connection by interpreting both twisting and the introduction of classical backgrounds as instances of twisting curved quantum \(L_\infty\)-superalgebras. Using the language of homotopy algebras and the Batalin–Vilkovisky formalism, we provide a unified algebraic framework that encompasses topological/holomorphic twists, spontaneous symmetry breaking, computation of anomalies, and supersymmetric localisation à la Festuccia--Seiberg. As a byproduct, we introduce a notion of twisting for quantum \(L_\infty\)-algebras
and a homotopy-algebraic reformulation of the one-particle-irreducible effective action.
\end{abstract}

\tableofcontents
\section{Introduction}
Two of the pillars of supersymmetric quantum field theory are the method of introducing classical backgrounds and the method of twisting.
The former underlies spontaneous symmetry breaking and the Higgs mechanism, where one introduces a classical background (nonzero vacuum expectation value) of a scalar field to manifest a true vacuum of the theory; when applied to the quantum theory, turning on a classical gauge field or metric background uncovers gauge and gravitational anomalies.
The latter isolates topologically or holomorphically protected sectors of supersymmetric gauge theories and sheds important light on string theory \cite{Witten:1988xj,Costello:2018zrm}, the (2,0) theory \cite{Beem:2014kka,pirsa_PIRSA:23070028}, M-theory \cite{Costello:2016mgj,Raghavendran:2021qbh,Cushing:2022uyd,Hahner:2024xrh,Hahner:2023kts,pirsa_PIRSA:23070028}, and beyond, in addition to being of fundamental importance to low-dimensional topology in the guise of Donaldson--Witten \cite{10.4310/jdg/1214437665,Witten:1988ze} and Seiberg--Witten \cite{Witten:1994cg} theories.

It is a curious fact that the two ingredients co-occur in the approach of Festuccia and Seiberg \cite{Festuccia:2011ws} (reviewed in \cite{Dumitrescu:2016ltq}) to supersymmetric localisation, in which one takes a supersymmetric theory, turns on some classical supergravity backgrounds, and
then twists using an unbroken supersymmetry generator.
This paper shows that this is no mere coincidence.
It has been recently recognised \cite{Elliott:2020ecf,Elliott:2020uwn} that twisting and classical backgrounds are intimately connected --- that, in a sense, twisting amounts to putting on a certain sort of global background field.
We sharpen this connection using the \(L_\infty\)-algebra formalism \cite{Jurco:2018sby,Arvanitakis:2019ald,Macrelli:2019afx,Jurco:2019yfd,Saemann:2020oyz,Jurco:2020yyu,Borsten:2021hua,Borsten:2021gyl,Borsten:2022ouu} for quantum field theory to cast both in terms of twisting quantum curved \(L_\infty\)-superalgebras. Using this  perspective, we uniformly formulate topological/holomorphic twisting, Higgsing, anomaly computations and supersymmetric localisation in the language of \(L_\infty\)-algebras and the Batalin--Vilkovisky formalism \cite{1983PhRvD..28.2567B,Batalin:1985qj,BATALIN1984106,Batalin:1981jr,Batalin:1977pb}.

Besides conceptual simplification, our uniform formulation clarifies various issues in the literature. Because we work consistently with homotopy algebras, it is clear that all our considerations apply straightforwardly to higher (e.g.~higher-form symmetries \cite{Gomes:2023ahz,Cordova:2022ruw,Brennan:2023mmt,Luo:2023ive,Bhardwaj:2023kri}) and/or nonstrict symmetries (e.g.~on-shell representations of supersymmetry \cite{Eager:2021wpi}). In particular, the existing literature on supersymmetric localisation \cite{Festuccia:2011ws,Dumitrescu:2016ltq} always works with off-shell representations of local supersymmetry including auxiliary fields; our perspective makes it clear that, at least in principle, this is not necessary.

As part of our discussion, we formulate the notion of twisting for quantum \(L_\infty\)-algebras, which is straightforward but nevertheless appears to be new to the literature.

\paragraph{Organisation of this paper.}
This paper is organised as follows. \Cref{sec:maths} reviews the twisting of (curved) \(L_\infty\)-algebras and generalises this construction to the quantum case, and \cref{sec:effective} reformulates the construction of one-particle-irreducible actions in the language of \(L_\infty\)-algebras. Then \cref{sec:classical_background} formulates perturbation theory atop a classical background --- including spontaneous symmetry breaking and the Higgs mechanism --- as a twist by a (quantum) Maurer--Cartan element corresponding to the background. \Cref{sec:anomaly} then applies this construction to one-particle-irreducible effective actions to detect anomalies.
\Cref{sec:twistingtwisting} reviews the physics notion of topological or holomorphic twists and their relation to the twists of \(L_\infty\)-algebras as `gauging' a global symmetry and turning on a constant `ghost' background. Finally, \cref{sec:localisation} discusses the supersymmetric localisation technique of Festuccia--Seiberg using supergravity backgrounds and its formulation in terms of \(L_\infty\)-algebra twists.

\paragraph{Notation.}
In this paper, we use the Koszul sign convention throughout. The notation \(V[i]\) indicates suspension, i.e.\ \((V[i])_j\coloneqq V_{i+j}\); where convenient, we also denote suspension by \(s\). The notation \(\odot\) denotes graded symmetrisation.

\section{Twisting classical and quantum \texorpdfstring{\(L_\infty\)}{L∞}-algebras}\label{sec:maths}
Twisting is a fundamental operation in the theory of classical and quantum \(L_\infty\)-algebras, which has manifold manifestations in field theory, as subsequent sections endeavour to prove.
The constructions are well known and classical for (classical) \(L_\infty\)-algebras \cite{dotsenko2019twisting,Dotsenko:2022bbv,Kraft:2022efy}, but appear to be new (while straightforward) for the quantum case.

\subsection{Classical \texorpdfstring{\(L_\infty\)}{L∞}-algebras and their twists}
We briefly review the relevant notions of \(L_\infty\)-superalgebras and their twists. For more detailed reviews, see \cite{Loday:2012aa,Jurco:2018sby,dotsenko2019twisting,Dotsenko:2022bbv,Kraft:2022efy}.

A curved \(L_\infty\)-superalgebra is the homotopy generalisation of the notion of a Lie superalgebra.\footnote{If one wants to work with theories containing fermions, working with superalgebras is necessary.} There are three  equivalent standard definitions of curved \(L_\infty\)-superalgebras, which are heuristically (we give  definitions below):
\begin{enumerate}
\item A $\ZZ\times\ZZ_2$-graded vector space $\mathfrak{g}$ with multilinear maps $\mu_i\colon\fg^{\times i}\to \fg$ that obey homotopy Jacobi relations. This picture  directly generalises the standard definition of a Lie superalgebra, where $\mathfrak{g}$ is a supervector space and the Lie bracket $[x_1, x_2]=\mu_2(x_1, x_2)$ is the only non-trivial bracket.
\item A cofree cocommutative $\ZZ\times\ZZ_2$-graded coalgebra with coderivation $D$ obeying $D^2=0$. The coderivation encodes the products $\mu_i$ and $D^2=0$ imposes the  homotopy Jacobi relations. This picture makes morphisms of algebras easy to state and most naturally captures the scattering amplitudes in quantum field theory.

\item A $\mathbb Z\times\mathbb Z_2$-graded manifold $M$ that is merely a $\mathbb Z\times\mathbb Z_2$-graded vector space endowed with a degree $(1,0)$ vector field $\Q$  satisfying $\Q^2=0$. This picture connects most directly to the Batalin--Vilkovisky formalism (as realised by symplectic $\mathbb Z\times\mathbb Z_2$-graded manifolds).
\end{enumerate}

Each picture offers its own advantages and we shall freely move between them as convenience dictates. In the following we briefly introduce the key definitions and connections amongst these formulations that we shall need throughout. For a detailed account see, for example, \cite{Jurco:2018sby}. 

First, we have the standard definition of an $L_\infty$-superalgebra, with explicitly stated   homotopy Jacobi   identities:

\begin{definition}\label{def:Linfty}
Over any commutative algebra \(\mathbb K\) over \(\mathbb F \supseteq\mathbb Q\),
a \emph{curved \(L_\infty\)-superalgebra} \((\mathfrak g,(\mu_i)_{i=0}^\infty)\) consists of a free \(\mathbb Z\times\mathbb Z_2\)-graded\footnote{In what follows, we refer to the $\ZZ$-grading as \emph{cohomological degree} and the $\ZZ_2$-grading as \emph{super degree}.
} \(\mathbb K\)-module \(\mathfrak g\) equipped with totally graded-antisymmetric \(\mathbb K\)-multilinear operations
\begin{equation}
    \mu_i\colon\bigwedge^i\mathfrak g\to\mathfrak g
\end{equation}
of degree \((2-i,0)\)
that satisfy the following homotopy Jacobi identity:
\begin{equation}
    \sum_{\mathclap{\substack{i+j=k\\\sigma\in\operatorname{Sym}(k)}}}\frac{(-1)^{ij}}{i!j!}\chi(\sigma)\mu_{j+1}\mleft(\mu_i(x_{\sigma(1)},\dotsc,x_{\sigma(i)}),x_{\sigma(i+1)},\dotsc,x_{\sigma(k)}\mright)=0,
\end{equation}
where the sum ranges over permutations of \(\{1,\dotsc,k\}\)
and where \(\chi(\sigma)\) is the graded-antisymmetric Koszul sign defined so that
\begin{equation}
    x_{\sigma(1)}\wedge\dotsb\wedge x_{\sigma(n)}\eqqcolon\chi(\sigma)x_1\wedge\dotsb\wedge x_n
\end{equation}
where
\begin{equation}
    x\wedge y=-(-1)^{pq+rs}y\wedge x
\end{equation}
for homogeneous elements \(x\) and \(y\) of degrees \((p,r)\) and \((q,s)\), respectively (\(p,q\in\mathbb Z\), \(r,s\in\{0,1\}\)).
\end{definition}

\begin{definition}\label{def:Linfty} A \emph{flat} \(L_\infty\)-superalgebra is a curved \(L_\infty\)-superalgebra with \(\mu_0=0\).
\end{definition}
Flat $L_\infty$-superalgebras are especially nice in the sense that the homotopy Jacobi relations implies $\mu_1{}^2=0$, so that they have  $\mu_1$-cohomology. (A curved $L_\infty$-superalgebra is, in the literature, sometimes also referred to as a \emph{weak} $L_\infty$-superalgebra since non-trivial $\mu_0$ implies $\mu_1$ is \emph{not} a differential.)

\begin{example} A curved \(L_\infty\)-superalgebra with \(\mu_i=0\) save for \(\mu_1\) and \(\mu_2\) is precisely equivalent to a differential graded Lie superalgebra, with \(\mu_2\) the Lie bracket and \(\mu_1\) the differential. We call such algebras \emph{strict} $L_\infty$-superalgebras.
\end{example}
\begin{example} A curved \(L_\infty\)-superalgebra concentrated in cohomological degree \(0\) is precisely equivalent to a Lie superalgebra.\end{example}
\begin{example} A curved \(L_\infty\)-superalgebra concentrated in degree \((0,0)\) is precisely equivalent to an (ordinary) Lie algebra.\end{example}

By operadic Koszul duality \cite{2007arXiv0709.1228G},
\cref{def:Linfty} is equivalent (see for example \cite{Loday:2012aa,Jurco:2018sby, Kraft:2022efy}) to a nilquadratic  degree $1$ coderivation $D$ on the counital cofree cocommutative coassociative coalgebra \(\bigodot\mathfrak g[1]\), where $\bigodot V$ is the graded-symmetric tensor coalgebra, and for a graded vector space $V$,
\begin{equation}
  V[k]\coloneqq [k]\otimes V,
\end{equation}
where $[k]$ denotes the one-dimensional vector space concentrated in degree $-k$, that is, $V[k]^{i}\cong  V^{i+k}$.
The coproduct on $\bigodot\fg[1]$ is given by 
\begin{equation}
\upDelta\left(x_1\odot \cdots \odot x_n\right)=\sum_{\mathclap{\substack{k=0\\\sigma\in\operatorname{Sh}(k, n-k)}}}^n  \chi(\sigma_{k,n})\left(x_{\sigma(1)}\odot  \cdots  \odot x_{\sigma(k)}\right) \otimes\left(x_{\sigma(k+1)}\odot  \cdots \odot  x_{\sigma(n)}\right), 
\end{equation}
where the sum is over $(k, n-k)$-shuffles $\sigma\in\operatorname{Sh}(k, n-k) \subset S_n$.  The products $\mu_i$ are packaged into the coderivation in the obvious manner (for details see \cite{Jurco:2018sby}),
\begin{equation}
\mu_i\coloneqq \pm s^{-1} \circ D_i \circ s^{\odot i},
\end{equation}
where  $D_i$ is the restriction of the coderivation to $\bigodot^i\mathfrak g[1]\to \mathfrak g[1]$ and $s^{\odot i}\colon \mathfrak \bigwedge^{i} \fg\to\bigodot^i \mathfrak g[1]$ is the isomorphism, 
\begin{equation}
s^{\odot i}(x_1 \wedge \ldots \wedge x_i )= (-1)^{\sum_{j=1}^{i-1}(i-j)\left|x_j\right|} s x_1 \odot \cdots \odot s x_i,
\end{equation}
 where $s\colon V\to V[1]$ denotes the suspension map. The nilquadraticity condition \(D^2=0\) then encodes the homotopy Jacobi relations amongst the $\mu_i$. This is equivalent to  the  \emph{Chevalley--Eilenberg coalgebra} on $\mathfrak g$ and so also denoted \(\CEc(\mathfrak g)\coloneqq(\bigodot\mathfrak g[1],D)\).

The dual  Chevalley--Eilenberg \emph{algebra} \(\CE(\mathfrak g)\coloneqq( \bigodot\mathfrak g[1]^*,D^*)\) is equivalent to a differential graded supermanifold,\footnote{Here we gloss over the difference between smooth functions versus polynomial functions if \(\mathfrak g[1]\) has nontrivial degree-zero elements.} that is, $C^\infty(\fg[1])$ together with a nilquadratic homological degree $1$ vector field $\Q$. The nilquadraticity condition $\Q^2=0$ is precisely equivalent to the classical master equation in the Batalin--Vilkovisky formalism: a curved \(L_\infty\)-superalgebra is nothing but a differential graded supermanifold whose underlying graded supermanifold is a graded supervector space. The other key ingredient of the Batalin--Vilkovisky formalism is a $\Q$-invariant  degree $-1$ symplectic form $\omega$ satisfying $i_\Q \omega= \mathrm d S$, where $S$ is identified with the classical Batalin--Vilkovisky action.

\begin{definition}[Cyclic structure]\label{def:cyc} A \emph{cyclic structure} \(\langle-,-\rangle_{\mathfrak g^*}\) on a curved \(L_\infty\)-super\-algebra \(\mathfrak g\) is a
degree $(3,0)$ graded symmetric bilinear pairing
\begin{equation}
\langle-,-\rangle_{\mathfrak g^*}\colon \fg^* \times \fg^* \rightarrow \mathbb{F}
\end{equation}
such that, picking a basis \(\{t_a\}\) of \(\mathfrak g\), if the structure constants of \(\mu_i\) are \(f^{a_0}{}_{a_1\dotso a_i}\) and the structure constants of \(\langle-,-\rangle_{\mathfrak g^*}\) are \(c^{ab}\), then
\begin{equation}
    f^{a_0\dotso a_i}\coloneqq f^{a_0}{}_{b_1\dotso b_i}c^{a_1b_1}\dotsm c^{a_ib_i}
\end{equation}
is totally graded-antisymmetric.
\end{definition}
In the above, we have permitted degenerate pairings to allow for examples such as the inner-derivation algebra (\cref{def:inner_derivation_algebra}). If \(\langle-,-\rangle_{\mathfrak g^*}\) is nondegenerate, it may be inverted to produce a degree \(-3\) pairing on \(\mathfrak g\), in terms of which one may phrase the definition as follows:
\begin{definition}[Nondegenerate cyclic structure] A \emph{nondegenerate cyclic structure} \(\langle-,-\rangle_{\mathfrak g}\) on a curved \(L_\infty\)-super\-algebra \(\mathfrak g\) is a
degree $(-3,0)$ graded symmetric non-degenerate bilinear pairing
\begin{equation}
\langle-,-\rangle_{\mathfrak g}\colon \fg \times \fg \rightarrow \mathbb{F}
\end{equation}
satisfying
\begin{equation}\label{eq:cyclicity}
\left\langle x_0, \mu_i\mleft(x_1, \cdots, x_{i}\mright)\right\rangle_{\mathfrak g}=(-1)^{i+i\left(\left|x_1\right|+\left|x_{i}\right|\right)+\left|x_{i}\right|\sum_{j=0}^{i-1}\left|x_j\right|}  \left\langle x_{i}, \mu_i\mleft(x_0, \dotsc, x_1\mright)\right\rangle_{\mathfrak g}
\end{equation}
 for homogeneous $x_0, \ldots, x_{i} \in \fg$.
\end{definition}
Note that a $\Q$-invariant (i.e.\ $\mathcal{L}_\Q \omega=0$) symplectic form \(\omega\) is equivalent to a nondegenerate structure on $\fg$.

\begin{definition} Let $(\fg,(\mu_i^\fg)_{i=0}^\infty)$ and $(\fh,(\mu_i^\fh)_{i=0}^\infty)$ be two curved $L_\infty$-superalgebras. An \emph{$L_\infty$-morphism of curved $L_\infty$-superalgebras} $\phi\colon\fg\rightsquigarrow \fh$ is a morphism of counital cocommutative differential graded coalgebras
\begin{equation}
  \phi\colon\CEc(\fg)\to \CEc(\fh), \qquad \phi\circ D_\fg = D_\fh \circ \phi.
\end{equation}
  
\end{definition}

\begin{definition}
A \emph{Maurer--Cartan element} \(Q\in\mathfrak g\) in a curved \(L_\infty\)-superalgebra with graded supervector space \(\mathfrak g=\bigoplus_{(i,j)\in\mathbb Z\times\mathbb Z_2}\mathfrak g^{i,j}\) is an element of degree \((1,0)\) such that the following sum exists and is zero:
\begin{equation}\label{eq:Maurer-Cartan}
    0 = \sum_{i=0}^\infty\frac1{i!}\mu_i(Q,\dotsc,Q).
\end{equation}
Equivalently,   on the counital cofree coalgebra \(\bigodot\mathfrak g[1]\) a Maurer--Cartan element $sQ\in \fg[1]^{(0,0)}\cong\fg^{(1,0)}$ satisfies 
\begin{equation}
 0 = D (\exp(sQ))|_{\mathfrak g[1]}=\sum_{i=0}^\infty\frac1{i!}D_i(sQ\odot sQ \dotsc \odot sQ).
\end{equation}
This definition immediately implies that $Q$ is a (curved) Maurer--Cartan element if and only if $D(\exp(sQ))=0$ \cite{Dolgushev:2005sp, Kraft:2022efy}. 
\end{definition}

If $\fg$ is equipped with a nondegenerate cyclic structure, then the metric \(\langle-,-\rangle_{\mathfrak g^*}\) on \(\mathfrak g^*\) may be inverted to yield a metric \(\langle-,-\rangle_{\mathfrak g}\) on \(\mathfrak g\) and  the Maurer--Cartan equation is variational in the sense that it extremises the Maurer--Cartan action
\begin{equation}\label{eq:hMC_definition}
S_{\mathrm{MC}}[x]\coloneqq\sum_{i} \frac{1}{(i+1)!}\left\langle x, \mu_i(x, \ldots, x)\right\rangle_{\mathfrak g},
\end{equation}
where $x$ has degree $(1,0)$. 
 If \(\langle-,-\rangle_{\mathfrak g^*}\) is degenerate, one may formally
define \(\langle-,-\rangle_{\mathfrak g}\) to be a (noncanonical) pseudoinverse of \(\langle-,-\rangle_{\mathfrak g^*}\), and  the expression \eqref{eq:hMC_definition} also exists,
although the variational relation to the Maurer--Cartan equation may not hold.

\begin{definition}[\cite{Kraft:2022efy}]
  \label{def:classicaltwist}
Given a curved \(L_\infty\)-superalgebra \(\mathfrak g\) and an element \(Q\in\mathfrak g^{1,0}\) of degree \((1,0)\), the \emph{twist} of \(\mathfrak g\) by \(Q\) is the curved \(L_\infty\)-superalgebra \((\mathfrak g_Q,(\mu_i^Q)_{i=0}^\infty)\) where \(\mathfrak g_Q  = \mathfrak g\) as graded supervector spaces, and 
\begin{equation}
    \mu^Q_i(x_1,\dotsc,x_i) \coloneqq \sum_{k=0}^\infty\frac1{k!} \mu_{i+k}(Q,Q,\dotsc,Q,x_1,\dotsc,x_i)
\end{equation}
if the above sum always exists.
Equivalently, on the counital cofree coalgebra \(\bigodot\mathfrak g[1]\) the twist is given by 
\begin{equation}
    D_Q = \exp(-\widehat{sQ})D\exp(\widehat{sQ}),
\end{equation}
where for $x, y\in \bigodot\mathfrak g[1]$, we define the operator $\hat x\colon\bigodot\mathfrak g[1]\to \bigodot\mathfrak g[1]$ by  $\hat x (y) = x \odot y$. 

\end{definition}

\begin{lemma}[ {\cite[Lem. 5.14 and Prop. 5.28]{Kraft:2022efy}} ]\label{lem:Qfacts} Let $\phi\colon(\CEc(\fg),D) \rightarrow\left( \CEc(\fg'),D'\right)$ be an $L_{\infty}$-morphism of (curved) $L_{\infty}$-superalgebras and denote by  $\phi^1$ its restriction to $\bigodot\mathfrak g[1] \to \mathfrak g'[1]$. Let \(Q \in \fg^{1,0}\). Then:

  \begin{enumerate}
    \item The twist $\fg_Q$ of $\fg$ by $Q$ is a flat $L_\infty$-superalgebra if and only if $Q$ is a Maurer--Cartan element of $\fg$.
\item $\phi(\exp(sQ))=\exp(sQ_{\phi})$, where $Q_{\phi}\coloneqq s^{-1}\phi^1(\sum_{i=1}^{\infty} \frac{1}{i!} (sQ)^{\odot i})$. 
\item  If $Q$  is a Maurer--Cartan element of $\fg$, then $Q_{\phi}$ is a Maurer--Cartan element of $\fg'$.
  \item The map
    \begin{equation}
      \phi_Q\coloneq \exp(-\widehat{sQ_\phi})\phi\exp(\widehat{sQ})\colon\CEc(\fg_Q)\to \CEc(\fg'_{Q_\phi})
    \end{equation}
    defines an $L_\infty$-morphism of curved $L_\infty$-superalgebras. 
\item If $\fg$ and $\fg'$ are flat ($\mu_0=0$) $L_\infty$-superalgebras, $Q\in\fg^{1,0}$ is a Maurer--Cartan element, and $\phi$ is a quasi-isomorphism, then $\phi_Q\colon\fg_Q\rightsquigarrow \fg'_{Q_{\phi}}$ is a quasi-isomorphism.
\end{enumerate}
\end{lemma}

\begin{theorem}[Minimal-model theorem]\label{thm:mm} Every flat $L_\infty$-superalgebra \((\mathfrak g,(\mu_i)_{i=1}^\infty)\)  is quasi-isomorphic to a (representative of the $L_\infty$-isomorphism class of) \emph{minimal model(s)} \((\mathfrak g^\circ,(\mu^\circ_i)_{i=2}^\infty)\), where \(\mathfrak g^\circ\cong H^\bullet_{\mu_1}(\mathfrak g)\) and \(\mu^\circ_1=0\). 
\end{theorem}

The twisting by a Maurer--Cartan element $Q$ is compatible with \autoref{thm:mm} in the following sense. Shifting to the coalgebra picture, let $\fg$ be a flat $L_\infty$-superalgebra, and let $\phi^\circ\colon \CEc(\fg)\to \CEc(\fg^\circ)$ denote the quasi-isomorphism to its minimal model $\fg^\circ$. By \cref{lem:Qfacts}, if $Q\in\fg^{1,0}$ is a Maurer--Cartan element, the map
\begin{equation}
\phi_Q^\circ \coloneqq  \exp(-\widehat{sQ_{\phi^\circ}})\phi^\circ \exp(\widehat{sQ}), 
\end{equation}
where
\begin{equation}
    Q_{\phi^\circ}=s^{-1}\phi^\circ{}^1(\sum_{i=1}^{\infty} \frac{1}{i!} (sQ)^{\odot i}),
\end{equation}
defines a quasi-isomorphism of  $L_\infty$-superalgebras 
\begin{equation}
  (\CEc(\fg),D_Q)\to (\CEc(\fg^\circ),D^\circ_{Q_{\phi^\circ}}).
\end{equation}
Since $Q_{\phi^\circ}$ is a  Maurer--Cartan element, the twisted minimal model is flat.
Moreover, $(\CEc(\mathfrak{g}_{Q_{\phi^\circ}}^\circ),D_{Q_{\phi^\circ}}^\circ)$ is a minimal model for $(\CEc(\mathfrak{g}_Q),D_Q )$.

\subsection{Quantum  \texorpdfstring{\(L_\infty\)}{L∞}-algebras and their twists}

Quantum \(L_\infty\)-algebras have their origins in closed string field theory \cite{Zwiebach:1992ie} and have been subsequently developed in e.g.~\cite{Markl:1997bj,Braun:2013lwa,Pulmannthesis,Doubek:2017naz, Jurco:2019yfd,Saemann:2020oyz,Doubek:2020rbg}, with the Batalin--Vilkovisky (BV) formalism \cite{1983PhRvD..28.2567B,Batalin:1985qj} often providing the motivation underlying the formal structures. In essence,  they generalise ordinary \(L_\infty\)-algebras in that they allow for an expansion in a formal parameter \(\hbar\) by imposing
the quantum master equation in addition to the classical master equation at zero loop order.

To identify the appropriate notions required for twists of quantum curved \texorpdfstring{\(L_\infty\)}{L∞}-algebras, the Batalin--Vilkovisky picture is instructive, so we briefly summarise the key points here. Detailed reviews may be found in \cite{Pulmannthesis,Doubek:2017naz, Jurco:2019yfd,Saemann:2020oyz,Doubek:2020rbg}. 

First, recall the classical master equation \(\{S_0,S_0\}=0\) for the classical Batalin--Vilkovisky action $S_0$ implies \(\Q_{\text{BV}}^2=0\), where $\Q_{\text{BV}}$ is the Batalin--Vilkovisky differential as defined by \(\Q_{\text{BV}}=\{S_0,-\}\) and $\{-,-\}$ is the Batalin--Vilkovisky antibracket.

When considering the partition function in the  Batalin--Vilkovisky formalism one must generalise the classical master equation \(\{S_0,S_0\}=0\) to the quantum master equation
\begin{equation}\label{eq:QME}
  (\Q_{\text{BV}}-2\mathrm i\hbar\Delta)S=  \{S,S\}-2\mathrm i\hbar\Delta S =0,
\end{equation}
where \(\Delta\) is the Batalin--Vilkovisky Laplacian\footnote{Properly speaking, the quantum master equation is a formal expression for local field theories due to the singular character of $\Delta$. However, \(\Delta\) can be rigorously defined in the case of a finite-dimensional field space, which can be implemented by e.g.\ using a periodic lattice of finite lattice spacing, which then provides ultraviolet and infrared regulators.} \cite{Batalin:1981jr}. This ensures that the expectation values of physical observables (gauge-invariant operators), which live in the cohomology of  $\Q_{\text{BV}}-2\mathrm i\hbar\Delta$, are independent of the choice of gauge.

 Solving the quantum master equation order-by-order in \(\hbar\) then produces the required ``counterterms'' of order \(\hbar^g\) that are to be added to the bare classical action \(S_0\), so that one has an \(\hbar\)-expansion of the quantum action
\begin{equation}
    S = S_0 + \hbar S_1 + \hbar^2 S_2+\dotsb,
\end{equation}
with \(S_0\) satisfying the classical master equation \(\{S_0,S_0\}=0\).  It may be that one can simply rescale the coefficients in \(S_0\) rather than add entirely new terms; in this case, the classical action already satisfies the quantum master equation. For instance, this is always the case in the absence of gauge symmetries, simply because the classical Batalin--Vilkovisky action has no antifield dependence so that $\Delta S_0=0$ from the get-go.  Then the usual \(L_\infty\)-algebraic perturbation theory continues to hold provided that one replaces \(D\mapsto D-\mathrm i\hbar\Delta^*\) where \(\Delta^*\) is the corresponding  \emph{dual} Batalin--Vilkovisky Laplacian in the coalgebra picture \cite{Jurco:2019yfd}. Their minimal models compute the loop-level scattering amplitudes\footnote{
    In order to define \(\Delta^*\) rigorously, one needs to regularise the \(L_\infty\)-algebra \(\mathfrak g\) to be finite-dimensional, by e.g.\ ultraviolet and infrared cutoffs, so that the quantum minimal model computes the regularised scattering amplitudes depending on the regularisation scales. From this one may perform the usual procedure of renormalisation by sending the regularisation scales to zero/infinity while renormalising the coupling constants to compensate.
} of quantum field theories \cite{Jurco:2019yfd}.\footnote{
    For a related and less amplitudes-oriented approach to renormalisation, see \cite{zbMATH05866383}, which works with local \(L_\infty\)-algebras.
}

This pictures leads one to the definition of a quantum curved $L_\infty$-algebra \cite{Zwiebach:1992ie, Markl:1997bj}:
\begin{definition}
  A quantum curved \(L_\infty\)-superalgebra \((\mathfrak g,(\mu^g_i)_{i,g=0}^\infty, \langle - , - \rangle)\) consists of
  \begin{enumerate}[\it(i)]
    \item a \(\mathbb Z\times\mathbb Z_2\)-graded \(\mathbb R\)-vector space \(\mathfrak g\),
    \item degree $(2-i,0)$ graded anti-symmetric $i$-linear maps $\mu^g_i\colon\bigwedge^i\mathfrak g\to\mathfrak g$, for each `genus' $g\geq0$, and
    \item a  degree $(3,0)$ graded-symmetric bilinear form 
    $\langle - , - \rangle \colon \mathfrak g^* \times \mathfrak g^*\to \mathbb{R}$ (or, dually, a homogeneous element of \(\bigodot\mathfrak g\) of word length \(2\) and degree \((3,0)\)),\footnote{
        If this is nondegenerate, we may invert it to produce a degree \((-3,0)\) graded-symmetric bilinear form \(\langle - , - \rangle \colon \mathfrak g \times \mathfrak g\to \mathbb{R}\). We do not assume nondegeneracy here, however, since it fails in the case of the inner-derivation algebra (\cref{def:quantum_inner_derivation_algebra}).
    }
  \end{enumerate}
  satisfying the following axioms (\cite{Zwiebach:1992ie, Markl:1997bj}).
  For a basis $\{t_a\}$ of $\fg$, with the bilinear form \(\langle-,-\rangle\) encoded as \(c^{ab}t_a \odot t_b\in \bigodot\mathfrak g\) for some structure constants \(c^{ab}\),
  we have, for any $n,g\geq0$ and $x_1,\ldots,x_n\in\fg$
  \begin{multline}
    \label{eq:main identity}
    0=\sum_{\mathclap{\substack{k+l=n+1\\g_1+g_2=g\\\sigma\in\operatorname{Sym}(n)}}} \frac{\chi(\sigma)(-1)^{l(k-1)}}{l!(k-1)!}\mu_{k}^{g_1}(\mu_{l}^{g_2}(x_{\sigma(1)},\ldots,x_{\sigma(l)}),x_{\sigma(l+1)},\ldots,x_{\sigma(n)})\\[-1.5em]
    +\frac1{2}\sum_{a}(-1)^{t_a+n}\mu_{n+2}^{g-1}(c^{ab}t_a,t_b,x_1,\ldots,x_n).
  \end{multline}
  Finally, the element
  \begin{equation}
    (-1)^{(n+1)t_a}c^{ab}t_a\otimes\mu_{n}^g(t_b, x_1,\ldots,x_n),
  \end{equation}
  is graded anti-symmetric for all $x_1,\ldots,x_n\in\fg$ and $n,g\geq0$. For $g=0$ this corresponds to cyclicity, see  \cref{def:cyc}.
\end{definition}

As for curved $L_\infty$-superalgebras, there is an equivalent and concise coalgebra definition of quantum curved $L_\infty$-superalgebras \cite{Markl:1997bj}, which encodes the main identity in a codifferential. Suppose we have the pair $(\mathfrak g, \langle - , - \rangle)$.  Then, to accommodate the genera of the products $\mu^g_i$ in the coalgebra picture, one first extends the tensor coalgebra,
\begin{equation}
\bigodot \mathfrak{g}[1]\llbracket\hbar\rrbracket\coloneqq \bigodot \mathfrak{g}[1]\otimes_\mathbb{R} \mathbb{R}\llbracket\hbar\rrbracket , 
\end{equation} where $\hbar$ is a formal parameter (of bidegree $(0,0)$) whose powers count the genus. There is a unique degree \((1,0)\)
second-order\footnote{A second-order coderivation $\theta$  satisfies 
\begin{equation}
\begin{aligned}
&(\upDelta \otimes \mathbb{1}) \circ \upDelta \circ \theta-\left(\mathbb{1}+\sigma_{231}+\sigma_{312}\right) \circ(\upDelta \otimes \mathbb{1}) \circ(\theta \otimes \mathbb{1}) \circ \upDelta \\
&+\left(\mathbb{1}+\sigma_{231}+\sigma_{312}\right) \circ(\theta \otimes \mathbb{1} \otimes \mathbb{1}) \circ(\upDelta \otimes \mathbb{1}) \circ \upDelta=0,
\end{aligned}
\end{equation}
where $\upDelta$ is the co-product on $\bigodot \mathfrak{g}[1]\llbracket\hbar\rrbracket$ (not to be confused with the Batalin--Vilkovisky Laplacian $\Delta$) and  $\sigma_{i j k}$ is the right action of the permutation group.  This is the dual of the seven-term  identity for second-order derivations as defined in \cite{AST_1985__S131__257_0,Akman:1995tm}. See \cite{Markl:1997bj} for details.} 
 coderivation
 \begin{equation}
\theta\colon \bigodot \mathfrak{g}[1]\llbracket\hbar\rrbracket \to \bigodot \mathfrak{g}[1]\llbracket\hbar\rrbracket
\end{equation}
such that 
\begin{equation}
  \pi_1\circ \theta=0, \quad \pi_2 \circ \theta(x)= \begin{cases}0, & x \in \bigodot^i \mathfrak{g}[1]\llbracket\hbar\rrbracket, i>0 \\ \frac{\hbar}{2} c^{ab}t_a \odot t_b, & x=1 \end{cases}
\end{equation}
where $\pi_i$ is the projector onto $\bigodot^i \mathfrak{g}[1]\llbracket\hbar\rrbracket$. These conditions further imply $\theta^2=0$. In the Batalin--Vilkovisky picture, $\theta$ will correspond to the dual of the Batalin--Vilkovisky Laplacian, so we will denote it by $\Delta^*$.

 Quantum curved $L_\infty$-superalgebra structures on \(\mathfrak g\) are then in bijection with degree \((1,0)\) coderivations $D$ on $\bigodot \mathfrak{g}[1]\llbracket\hbar\rrbracket$ such that
 \begin{equation}\label{eq:quantum_master_equation}
    (D+\hbar\Delta^*)^2 = 0
\end{equation}
as a formal power series in \(\hbar\). Switching to the dual graded derivation (i.e.~Batalin--Vilkovisky) picture, such coderivations are in bijection to solutions $S$ to the quantum master equation \eqref{eq:QME}, where $D$ and $\Delta^*$ are dual to $Q_{\text{BV}}$ and $\Delta$.

Given this, one may straightforwardly define the quantum generalisations of  Maurer--Cartan elements, twists and their compatibility with minimal models.

\begin{definition}A \emph{quantum Maurer--Cartan element} in a quantum curved \(L_\infty\)-superalgebra \((\fg, (\mu_i^g)_{i,g=0}^\infty, \Delta^*)\) is a family $\{Q^g\in \fg^{(1,0)}\}_{g=0}^\infty$ of elements of bidegree $(1,0)$,  such that, in the coalgebra picture, the following sum exists and is zero:
\begin{equation}\label{eq:coquantumMaurer-Cartan} 
 0 = (D +\hbar\Delta^*) \exp(s \sum_{g\geq0}\hbar^gQ^g).
\end{equation}
We will adopt the same notation as in the classical setting and let $Q\coloneqq \sum_{g\geq0}\hbar^gQ^g$.
\end{definition}
Note that, since \(\Delta^*\) always increases the length of words by two,
when restricting \eqref{eq:coquantumMaurer-Cartan} to words of length one, \(\Delta^*\) does not contribute, so that we have
\begin{equation}\label{eq:classical_MC_follows_from_quantum_MC}
     0 = \left.D\exp(s Q)\right|_{\mathfrak g[1]\llbracket\hbar\rrbracket}
     = \sum_{i,g=0}^\infty\frac1{i!}\hbar^g\mu^g_i(Q,\dotsc,Q).
\end{equation}
This has the same form as the classical Maurer--Cartan equation \eqref{eq:Maurer-Cartan}.

\begin{definition}
Given a quantum curved \(L_\infty\)-superalgebra \((\mathfrak g,(\mu_i^g)_{i,g=0}^\infty,\Delta^*)\) and
\begin{equation}Q\in\mathfrak g^{1,0}\llbracket\hbar\rrbracket,\end{equation}
the \emph{twist} by \(Q\) is the quantum curved \(L_\infty\)-superalgebra whose underlying graded vector space is that of \(\mathfrak g\) but whose operations are encoded by
\begin{equation}
    D_Q=\exp(-\widehat{sQ})D\exp(\widehat{sQ}).
\end{equation}
The Batalin--Vilkovisky Laplacian stays unchanged.
\end{definition}
\begin{lemma}
The twist of a quantum curved \(L_\infty\)-superalgebra by an element, if it exists, is indeed a quantum curved \(L_\infty\)-superalgebra.
\end{lemma}
\begin{proof}
First, we note that
\begin{equation}
    \Delta^* = \exp(-\widehat{sQ})\Delta^*\exp(\widehat{sQ}),
\end{equation}
or equivalently
\begin{equation}
    \exp(\widehat{sQ})\Delta^* = \Delta^*\exp(\widehat{sQ}),
\end{equation}
since the effect of \(\Delta^*\) is to simply insert pairs \(c^{ab}t_at_b\), which then commute with \(\exp(\widehat{sQ})\).

We must check the quantum master equation \eqref{eq:quantum_master_equation}, the nilquadraticity of \(\Delta^*\), as well as the second-order coderivation property of \(\Delta^*\).

The first and second are immediate; given any map $X$ such that $X^2=0$, it follows
\begin{equation}
    X_Q^2= \exp(-\widehat{sQ})X\exp(\widehat{sQ})\exp(-\widehat{sQ})X\exp(\widehat{sQ})=0.
\end{equation}
The third is also clear since the map
\begin{equation}
    x\mapsto \exp(-\widehat{sQ})x\exp(\widehat{sQ})
\end{equation}
is an automorphism of counital coalgebras and hence preserves the second-order property.
\end{proof}

Given a (non-curved) quantum \(L_\infty\)-superalgebra \((\mathfrak g,(\mu_i^g)_{i=1}^\infty,\Delta^*)\), then by the homological perturbation lemma one can define its \emph{minimal model} \cite{Doubek:2017naz,Jurco:2019yfd}, which is a quantum \(L_\infty\)-superalgebra \((\operatorname H(\mathfrak g),((\mu^\circ)_i^g)_{i=1}^\infty,(\Delta^*)^\circ)\) whose underlying graded vector space is \(\operatorname H(\fg)\coloneq\operatorname H_{\mu_1^0}(\mathfrak g)\), the cohomology of $\mu_1^0$.

Note that, in the special case where \(\mathfrak g\) describes a physical theory, the scattering amplitudes must only involve fields rather than antifields, so that, apart from the Batalin--Vilkovisky Laplacian \(\Delta^*\) the only nonzero structure map of the quantum minimal model \(\operatorname H(\mathfrak g)\) is \(\mu_i^g\colon\mathfrak g^1\otimes\dotso\otimes\mathfrak g^1\to\mathfrak g^2\). This means that the homotopy Maurer--Cartan action \(S_{\operatorname H(\mathfrak g)}\) of \(\operatorname H(\mathfrak g)\) does not depend on antifields, and so in particular its Batalin--Vilkovisky Laplacian vanishes:
\begin{equation}\label{eq:vanishing-laplacian-for-minimal-model}
    \Delta S_{\operatorname H(\mathfrak g)} =0.
\end{equation}
Therefore, if \(S_{\operatorname H(\mathfrak g)}\) satisfies the quantum master equation, it also satisfies the classical master equation. In other words, forgetting \(\Delta^*\), then \(\operatorname H(\mathfrak g)\) may be regarded as a cyclic (non-quantum) \(L_\infty\)-superalgebra over \(\mathbb R\llbracket\hbar\rrbracket\).

\section{Effective actions in the language of \(L_\infty\)-algebras}\label{sec:effective}
For physics applications (e.g.\ Higgsing, anomalies) we need to formulate the physical concept of one-particle-irreducible effective actions in terms of \(L_\infty\)-algebras.
The defining quality of the one-particle-irreducible effective action is such that its tree amplitudes equal the loop amplitudes of the original theory.
This may be phrased in terms of \(L_\infty\)-algebras as follows.
Given a quantum \(L_\infty\)-superalgebra \(\mathfrak g\) (the original theory) whose quantum minimal model is \(\operatorname H(\mathfrak g)\), suppose that \eqref{eq:vanishing-laplacian-for-minimal-model} holds so that \(\operatorname H(\mathfrak g)\) may also be regarded as a (non-quantum) cyclic \(L_\infty\)-superalgebra over \(\mathbb R\llbracket\hbar\rrbracket\). Then a \emph{one-particle-irreducible effective theory} is a (non-quantum) cyclic \(L_\infty\)-superalgebra \(\mathfrak g_{\mathrm{1PI}}\), whose underlying graded supervector space coincides with that of \(\mathfrak g\), and whose (non-quantum) minimal model \(\operatorname H(\mathfrak g_{\mathrm{1PI}})\) is isomorphic as a cyclic \(L_\infty\)-superalgebra over \(\mathbb R\llbracket\hbar\rrbracket\) to \(\operatorname H(\mathfrak g)\).

It is not obvious that such a \(\mathfrak g_{\mathrm{1PI}}\) exists. When \(\mathfrak g\) represents a physical quantum field theory, however, then \(\mathfrak g_{\mathrm{1PI}}\) is known to exist, either by an abstract construction in terms of a Legendre transformation of the free energy, or by an explicit diagrammatic construction. Both constructions may be phrased in terms of \(L_\infty\)-superalgebras as follows.
\begin{equation}
  \label{eq:sequence}
      \begin{tikzcd}
        &\fg \ar[dl,dashed,"{\text{Add sources } J}"'] \ar[dr,rightsquigarrow,"{\text{Quantum min.\ model}}"]&\\
        \operatorname{inn}(\fg)\ar[d,dashed,"{\text{Legendre transf.}}"']&&\operatorname H(\fg)\ar[d, dashed, "{\text{ Forget } \Delta^*}"]\\
        \fg_{1PI}\ar[dr, rightsquigarrow,"\text{Classical min.\ model}"']&&\operatorname H(\fg)\ar[dl,rightsquigarrow, "{L_\infty\text{-isomorphism}}"]\\
        &\operatorname H(\fg_{1PI})&
      \end{tikzcd}
    \end{equation}
    In the diagram \eqref{eq:sequence} the squiggly arrows represent equivalences of classical or quantum $L_\infty$-superalgebras, and the dashed ones are operations performed on (quantum) $L_\infty$-superalgebras, producing new ones.

\subsection{Abstract construction of the effective action}
Abstractly, the construction of the one-particle-irreducible effective action consists of the following steps:
\begin{enumerate}
\item To every field \(\phi_i\) in the action \(S[\phi]\), one adds a source \(J^i\) to obtain the source-extended action
\begin{equation}\label{eq:source-extended action}
    S'[\phi,J] = S[\phi]+ \int \phi_iJ^i.
\end{equation}
\item\label{item:step2} One takes the partition function by integrating out \(\phi\), leaving a functional \(W[J]\) of \(J\) alone:
\begin{equation}
    \exp(-\mathrm iW[J]) = \int\mathrm D\phi\exp(\mathrm iS'[\phi,J]).
\end{equation}
\item\label{item:step3} Finally, one takes the Legendre transform to obtain the one-particle-irreducible effective action
\begin{equation}
    \Gamma(\phi_\mathrm{cl})=\sup_J\left(-W[J]-\int J^i\phi_{\mathrm{cl},i}[J]\right),
\end{equation}
where
\begin{equation}\label{def:antifield-classical}
    \phi_{\mathrm{cl},i}[J]=\langle  \phi \rangle_J= \int\mathrm D\phi ~\phi \exp(\mathrm iS'[\phi,J]).
\end{equation}
\end{enumerate}
At least formally, we may phrase the above steps, generalised to the Batalin--Vilkovisky setting \cite{Batalin:1981jr},  in the language of \(L_\infty\)-algebras.

The first step, namely to introduce source terms into the action, corresponds to taking the inner-derivation algebra \(\operatorname{inn}(\mathfrak g)\) \cite{Sati:2008eg}, whose associated Chevalley--Eilenberg algebra is the Weil algebra, reviewed in e.g.\ \cite[§3.2]{Borsten:2024gox}.
This formally introduces sources for all elements of \(\mathfrak g\) (including ghosts and antifields),
but there are no antifields for the source, such that the cyclic structure is necessarily degenerate.
(See~\cite{Batalin:2013xpa} for the necessity of introducing sources for antifields in the Batalin--Vilkovisky formalism.)
Let us recall the definition.
\begin{definition}\label{def:inner_derivation_algebra}
Suppose that \((\mathfrak g,\langle-,-\rangle_{\mathfrak g^*})\) is an cyclic \(L_\infty\)-superalgebra with
the \(L_\infty\)-algebra structure given by a codifferential \(D_{\mathfrak g}\) on \(\bigodot\mathfrak g[1]\). The \emph{inner-derivation algebra} is the cyclic \(L_\infty\)-superalgebra
\begin{equation}
    \operatorname{inn}(\mathfrak g)\coloneqq\mathfrak g\oplus\mathfrak g[1]
\end{equation}
and with the differential \(D_{\operatorname{inn}(\mathfrak g)}\) on 
\begin{equation}
    \bigodot\operatorname{inn}(\mathfrak g)[1]
    =\bigodot\mathfrak g[1]
    \otimes
    \bigodot\mathfrak g[2]
\end{equation}
given by
\begin{equation}\label{eq:deltait}
    D_{\operatorname{inn}(\mathfrak g)} = 
    D_{\mathfrak g}
    + 
    \deltait,
\end{equation}
where \(\deltait\) is the degree-shift map that maps \(\mathfrak g[2]\) to \(\mathfrak g[1]\) inside \(\bigodot\operatorname{inn}(\mathfrak g)[1]\) (and annihilates \(\mathfrak g[1]\)) and extended by the graded Leibniz rule to the entirety of \(\bigodot\operatorname{inn}(\mathfrak g)[1]\), and
the action of \(D_{\mathfrak g}\) is extended from generators in \(\mathfrak g[1]\) to generators in \(\mathfrak g[2]\) via the commutation rule
\begin{equation}
    D_{\mathfrak g}\deltait=
    -\deltait D_{\mathfrak g}.
\end{equation}
The cyclic structure is given by
\begin{equation}\label{eq:inn_g_cyclic}
    \langle x+y,x'+y'\rangle_{\operatorname{inn}(\mathfrak g)^*}
    =
    \langle x,x'\rangle_{\mathfrak g^*}
\end{equation}
for \(x,x'\in\mathfrak g^*\) and \(y,y'\in\mathfrak g[1]^*\).
\end{definition}
The underlying vector space of the inner-derivation algebra may be interpreted as
\begin{equation}
    \operatorname{inn}(\mathfrak g)=\overbrace{\mathfrak g}^{\mathclap{\text{fields/antifields}}}\oplus\underbrace{\mathfrak g[1]}_{\mathclap{\text{sources}}}.
\end{equation}
Note that this definition formally introduces `sources' for every element in \(\mathfrak g\), including for antifields, unlike the usual discussion \cite{Batalin:1981jr}. This is at least harmless for gauge-fixed actions, where antifield-containing terms have been excised from the action.

Note that the cyclic structure \eqref{eq:inn_g_cyclic} is degenerate and fails to invert to produce a pairing on \(\operatorname{inn}(\mathfrak g)\) since there are no antifields for the sources.\footnote{
    If one defines a degenerate inner product on \(\operatorname{inn}(\mathfrak g)\), it will generally fail to obey the cyclic identity \eqref{eq:cyclicity}.
    Consider, for instance, a scalar field theory with action \(\int\phi\square\phi+J\phi\) with a source term \(\phi J\),
    and consider the degenerate metric pairing \(\phi\) with \(\phi^+\) and \(J\) with nothing.
    One then has \(\langle\phi,\mu_1(J)\rangle\sim\langle\phi,\phi^+\rangle\ne0\) whereas \(\langle\mu_1(\phi),J\rangle\sim\langle\phi^+,J\rangle=0\), violating \eqref{eq:cyclicity}.
}
(It is also difficult to add antifields for sources and then extend the \(L_\infty\)-superalgebra structure to sources in a suitable way.)
For \(\operatorname{inn}(\mathfrak g)\), if the cyclic structure on \(\mathfrak g\) is nondegenerate,
there exists a canonical pseudoinverse \(\langle-,-\rangle_{\operatorname{inn}(\mathfrak g)}\) of \(\langle-,-\rangle_{\operatorname{inn}(\mathfrak g)^*}\) given by
\begin{equation}
    \langle(x,y),(x',y')\rangle_{\operatorname{inn}(\mathfrak g)}
    =\langle x,x'\rangle_{\mathfrak g}
\end{equation}
for \(x,x'\in\mathfrak g\) and \(y,y'\in\mathfrak g[1]\),
and with respect to this structure one can define the homotopy Maurer--Cartan action \eqref{eq:hMC_definition} for \(\operatorname{inn}(\mathfrak g)\). The result of the additional \(\deltait\) term then produces, in the corresponding homotopy Maurer--Cartan action,  quadratic  source terms \(\phi_iJ^i\) in \eqref{eq:source-extended action}. The additional terms coming from extending \(D_{\mathfrak g}\) are needed for gauge invariance.

The above definition extends straightforwardly to the quantum case:
\begin{definition}\label{def:quantum_inner_derivation_algebra}
Suppose that \(\mathfrak g\) is a quantum \(L_\infty\)-superalgebra with
the quantum \(L_\infty\)-algebra structure given by the pairing \(\langle-,-\rangle_{\mathfrak g^*}\) and a codifferential \(D_{\mathfrak g}\) on \(\bigodot\mathfrak g[1]\llbracket\hbar\rrbracket\). The \emph{inner-derivation algebra} is
the quantum \(L_\infty\)-superalgebra
\begin{equation}
    \operatorname{inn}(\mathfrak g)\coloneqq\mathfrak g\oplus\mathfrak g[1]
\end{equation}
with the degenerate cyclic pairing
\begin{equation}
    \langle(x,y),(x',y')\rangle_{\operatorname{inn}(\mathfrak g)^*}
    =\langle x,x'\rangle_{\mathfrak g^*}
\end{equation}
for \(x,x'\in\mathfrak g^*\) and \(y,y'\in\mathfrak g[1]^*\),
and with the differential \(D_{\operatorname{inn}(\mathfrak g)}\) on 
\begin{equation}
    \bigodot\operatorname{inn}(\mathfrak g)[1]\llbracket\hbar\rrbracket
    =\bigodot\mathfrak g[1]\llbracket\hbar\rrbracket
    \otimes
    \bigodot\mathfrak g[2]\llbracket\hbar\rrbracket
\end{equation}
given by
\begin{equation}
    D_{\operatorname{inn}(\mathfrak g)} = 
    D_{\mathfrak g}
    + 
    \deltait,
\end{equation}
where \(\deltait\) is the degree-shift map that maps \(\mathfrak g[2]\llbracket\hbar\rrbracket\) to \(\mathfrak g[1]\llbracket\hbar\rrbracket\) inside \(\bigodot\operatorname{inn}(\mathfrak g)[1]\) (and maps \(\mathfrak g[1]\llbracket\hbar\rrbracket\) to zero) and extended by the graded Leibniz rule to the entirety of \(\bigodot\operatorname{inn}(\mathfrak g)[1]\llbracket\hbar\rrbracket\), and
the action of \(D_{\mathfrak g}\) is extended from generators in \(\mathfrak g[1]\llbracket\hbar\rrbracket\) to generators in \(\mathfrak g[2]\llbracket\hbar\rrbracket\) via the commutation rule
\begin{equation}
    0=D_{\mathfrak g}\deltait
    +\deltait D_{\mathfrak g}
    =\Delta^*\deltait+\deltait\Delta^*,
\end{equation}
where \(\Delta^*\) is the Batalin--Vilkovisky Laplacian encoding \(\langle-,-\rangle\).
\end{definition}

Thus, starting with a quantum \(L_\infty\)-superalgebra \(\mathfrak g\) describing a quantum field theory with a nondegenerate cyclic structure, we can construct  the homotopy Maurer--Cartan action,\footnote{
    In \eqref{eq:actionweil},
    the only additional terms are the source terms \(\langle\Phi,J\rangle\), arising from the \(\deltait\) in \eqref{eq:deltait}.
    The extension of \(D_{\operatorname{CE}(\mathfrak g)}\) to \(\bigodot\mathfrak g[2]^*\)
    does not yield additional terms in the action:
    the resulting terms in \(\mu_i\) are always valued in the shifted elements \(\mathfrak g[1]\subset\operatorname{inn}(\mathfrak g)\) rather than 
    \(\mathfrak g\subset\operatorname{inn}(\mathfrak g)\), and since \(\mathfrak g[1]\) lies in the kernel of \(\langle-,-\rangle_{\operatorname{inn}(\mathfrak g)}\) the resulting terms in \(\mu_i\) vanish inside \eqref{eq:hMC_definition}. 
}
\begin{equation}\label{eq:actionweil}
 S_{\operatorname{inn}(\mathfrak g)}[\Phi, J]= S[\Phi] +\langle \Phi, J\rangle,
 \end{equation} of \(\operatorname{inn}(\mathfrak g)\), which is now an (\(\hbar\)-power-series-valued) functional on the differential graded manifold \(\operatorname{inn}(\mathfrak g)[1]\), whose underlying graded space is \(\operatorname{inn}(\mathfrak g)[1]\cong\mathfrak g[1]\oplus\mathfrak g[2]\) and \(\Phi\in\mathfrak g[1]\) (original fields and antifields). The role of the  sources for the antifields (a subset of the $J$) is to merely to maintain manifest Batalin--Vilkovisky invariance, i.e.~solve the quantum master equation for $\operatorname{inn}(\fg)$.

Given the above,  we may now perform steps~\ref{item:step2} and \ref{item:step3} as usual, at least formally.
First we must gauge-fix; this distinguishes a choice of antifields  within the $L_\infty$-superalgebra $\fg$ so that we may write $\Phi = (\phi, \phi^+)$. Then, for a fixed configuration of  classical expectation values of the antifields $\phi^+_{\mathrm{cl}}$ as given in \eqref{def:antifield-classical}, we  define the analogue of the usual generating functional of connected correlations functions with sources $J=(J_\phi, J_{\phi^+})$, 
\begin{equation}
    \exp(-\mathrm iW[\phi^+_{\mathrm{cl}}, J])=\int\mathrm D\Phi\,\delta\left(\phi^+-\frac{\partial( \Psi + \phi^+_{\mathrm{cl}} \phi)}{\partial \phi}\right) \exp\mleft(\mathrm iS_{\operatorname{inn}(\mathfrak g)}\mright)
    \end{equation}
    where $\Psi$ is the gauge-fixing fermion (that identifies the Lagrangian submanifold of fields $\phi$). 
    
    The classical expectation value of the fields is then defined  by
    \begin{equation}
    \phi_{\mathrm{cl}} =\left. \frac{\partial W}{\partial J_\phi}\right|_{J_\phi=0}.
    \end{equation}
    To obtain a functional (the effective action) \(\Gamma(\Phi_\mathrm{cl})\) as a function of \(\Phi_\mathrm{cl}\), one takes the Lagendre transform as usual after setting the extraneous fields \(J_{\phi^+}\) to zero:
    \begin{equation}
    \Gamma(\Phi_\mathrm{cl})
      =\sup_{J_{\phi}}\left(-\left.W[\phi^+_{\mathrm{cl}}, J]\right|_{J_{\phi^+}=0}-\langle J_\phi, \phi_\mathrm{cl}\rangle\right)
\end{equation}
This reproduces the Batalin--Vilkovisky effective action of \cite{Batalin:1981jr,Gomis:1994he}, which satisfies the classical master equation in the space of classical (anti)fields, i.e.~the Zinn-Justin equation. Thus \(\Gamma\) is the homotopy Maurer--Cartan action of a cyclic \(L_\infty\)-superalgebra over \(\mathbb R\llbracket\hbar\rrbracket\), which we denote as \(\mathfrak g_{\mathrm{1PI}}\).

\begin{example}Consider $\phi^4$ scalar field theory:
\begin{equation}
    S_{\mathfrak g}=\int\frac12\phi\square\phi + \frac{\lambda}{4!}\phi^4 .
\end{equation}
The corresponding quantum \(L_\infty\)-algebra \(\mathfrak g\) is concentrated in degrees \(1\) and \(2\) for the scalar field \(\phi\) and its antifield \(\phi^+\) respectively and the only non-trivial brackets are 
\begin{equation}
\mu^0_1(\phi) = \Box \phi, \qquad \mu^0_3(\phi, \phi, \phi) = \phi^3.
\end{equation}
Then the corresponding \(\operatorname{inn}(\mathfrak g)=\mathfrak g\oplus\mathfrak g[1]\) is concentrated in \(1\) (field \(\phi\)), \(2\) (antifield \(\phi^+\)), \(0\) (source \(J_{\phi^+}\) for \(\phi^+\)), and \(1\) (source \(J_\phi\) for \(\phi\)), respectively. Solutions of the  classical master equation satisfy the  quantum master equation\footnote{Since the $J, J_{\phi^+}$ are independent of $\phi, \phi^+$, $ S_{\operatorname{inn}(\mathfrak g)}$ is trivially in the kernel of the Batalin--Vilkovisky Laplacian.} so the corresponding Maurer--Cartan action is
\begin{equation}
    S_{\operatorname{inn}(\mathfrak g)}[\phi, \phi^+,J_\phi,J_{\phi^+}]=\int\frac12\phi\square\phi+ \frac{\lambda}{4!}\phi^4+\phi J_\phi+\phi^+J_{\phi^+},
\end{equation}
where the final  two terms are the the expected   source couplings.

In this case the gauge fixing is trivial and thus
\begin{equation}
    \exp(-\mathrm iW[\phi^+_{\mathrm{cl}}, J])=\int\mathrm D\phi\,\exp\mathrm i\  S_{\operatorname{inn}(\mathfrak g)}[\phi, \phi_{\mathrm{cl}}^+, J_\phi,J_{\phi^+}].
\end{equation}
Then 
\begin{equation}
\Gamma[\phi_{\mathrm{cl}}, \phi^+_{\mathrm{cl}}]= \sup_{J_{\phi}}\left(-W[\phi^+_{\mathrm{cl}}, J_\phi]
-\langle J_\phi, \phi_\mathrm{cl}\rangle\right),
\end{equation}
where $W[\phi^+_{\mathrm{cl}}, J_\phi]= W[\phi^+_{\mathrm{cl}}, J]|_{J_{\phi^+}=0}$,  is precisely the Batalin--Vilkovisky effective action solving  the classical master equation. Consequently, $\Gamma[\phi_{\mathrm{cl}}, \phi^+_{\mathrm{cl}}]$ defines a classical (flat) $L_\infty$-algebra, whose minimal model encodes the loop amplitudes of the original model.
\end{example}

\subsection{Diagrammatic construction}
Rather than going through the formal construction, we may directly implement the diagrammatic construction of the quantum effective action \cite[(11.63)]{Peskin:1995ev}, that is,
\begin{equation}\label{eq:effective_action}
    \Gamma = S + \frac12\mathrm i\hbar\operatorname{tr}\ln\left(-\frac{\deltaup^2S}{(\deltaup\phi)^2}\right)
    -\mathrm i\sum\hbar^L\text{(\(L\)-loop 1PI conn.\ diag.)}.
\end{equation}
(This does not implement the convexity property directly, but this point is not relevant for our purposes.)

Suppose that we are given a
finite-dimensional quantum \(L_\infty\)-superalgebra \(\mathfrak g\), such as one obtained from a gauge-fixed quantum field theory action with suitable infrared and ultraviolet cutoffs so as to render the field space finite-dimensional.
We then consider formally an \(L_\infty\)-superalgebra structure over \(\mathbb R[\![\hbar]\!]\) on the underlying graded supervector space of \(\mathfrak g\otimes\mathbb R[\![\hbar]\!]\) given by
\begin{multline}
    \tilde\mu_i(x_1,\dotsc,x_n)
    =\\
    \sum_{g=0}^\infty \hbar^g\left(\mu_i^g(x_1,\dotsc,x_n)
    +\tilde\mu_i^{g,\det}(x_1,\dotsc,x_n)
    +\tilde\mu_i^{g,\mathrm{conn}}(x_1,\dotsc,x_n)\right),
\end{multline}
where \(\tilde\mu_i^{g,\det}\) and \(\tilde\mu_i^{g,\mathrm{conn}}\) implement the formula \eqref{eq:effective_action} (described below).

In detail, the term \(\tilde\mu_i^{g,\det}\) is defined by the structure constant \(\tilde c^g_{a_1\dotso a_n}\), whose definition
\begin{align}
    \kappa^g_{ab;c_1\dotso c_n}
    &=
    c^g_{abc_1\dotso c_n}
    -\deltaup_{g0}c^0_{ab}\notag\\
    \tilde c^g_{a_1\dotso a_n}&=\frac1p(-1)^{p+1}\binom n{n_1,\dotsc,n_p}\\
    &\quad\times\sum_{\mathclap{\substack{n_1+\dotsb+n_p=n\\g_1+\dotsb+g_p=g-1}}}\kappa^{g_1}_{b'_1b_2;a_1\dotso a_{n_1}}\hat c^{b_2b_2'}
    \kappa^{g_2}_{b'_2b_3;a_{n_1+1}\dotso a_{n_1+2}}\hat c^{b_3b_3'}\dotsm \kappa^{g_p}_{b_pb_1;a_{n_1+\dotsb+n_{p-1}+1}\dotso a_n}\hat c^{b_1b_1'}\notag
\end{align}
implements the Taylor series expansion \(\ln(a+x)=\sum_{p=1}^\infty p^{-1}(-1)^{p+1}a^{-p}x^p+\ln(a)\) (with the cosmological-constant term \(\ln(a)\) dropped), and where \(c^g_{a_1\dotso a_i}\) are the structure constants of \(\mu^g_i\), and where \(\hat c^{ab}\) is the inverse of \(c^0_{ab}\).

The term \(\tilde\mu_k^{g,\mathrm{conn}}\) is implemented as follows.
Given a pseudograph (i.e.\ undirected graph with self-loops \(\Feyn{fflof}\), parallel edges \(\maxis{\Feyn{ff0flfluf0f}}\), and external legs allowed) \(G\) that corresponds to a one-particle-irreducible Feynman diagram, let \(\operatorname{HE}(G)\) be its set of half-edges. (A self-loop still corresponds to two different half-edges, but an external leg corresponds to only one half-edge.) Let the indices \(\alpha,\beta,\dotsc\) range over \(\operatorname{HE}(G)\); each internal edge \(e\in\operatorname E(G)\) may then be represented by an unordered pair of half-edges.
For each pseudograph \(G\) with \(v\) vertices of degrees \(d_1,\dotsc,d_v\) and \(L\) loops and \(k\) external legs (so that there are \(d_1+\dotsb+d_v\) half-edges),
define the expression
\begin{multline}
    \tilde c^{g,G;a_{d_1+\dotsb+d_v+1}\dotso a_{d_1+\dotsb+d_v+k}}
    = -\mathrm i\frac1{|\operatorname{Aut}(G)|}\prod_{\{a_r,a_s\}\in\operatorname E(G)}\deltaup^{a_ra_s}\\\times\sum_{\mathclap{g_1+\dotsb+g_v=g-L}}c^{g_1}_{a_1\dotso a_{d_1}}c^{g_2}_{a_{d_1+1}\dotso a_{d_1+d_2}}\dotsm c^{g_v}_{a_{d_1+\dotsb+d_{v-1}+1}\dotso a_{d_1+\dotsb+d_v}},
\end{multline}
where \(\operatorname E(G)\) is the set of edges regarded as unordered pairs of half-edges; for external legs, \(\operatorname E(G)\) pairs up the external legs with the free indices \begin{equation}a_{d_1+\dotsb+d_v+1},\dotsc,a_{d_1+\dotsb+d_v+k}.\end{equation} The factor \(|\operatorname{Aut}(G)|\) is the order of the automorphism group of \(G\) (fixing external legs). Then \(\tilde\mu_k^{g,\mathrm{conn}}\) is defined by the structure constant
\begin{equation}
    \tilde c^{g,\mathrm{conn}}_{a_0a_1\dotso a_k}\coloneqq\sum_G\tilde c^{g,G}_{(a_0a_1\dotso a_k)},
\end{equation}
where the parentheses indicate normalised graded symmetrisation and where the sum ranges over the connected one-particle-irreducible pseudographs \(G\) with \(k+1\) external legs, and where we have lowered the indices using the cyclic structure.

\subsection{One-particle-irreducible Maurer--Cartan element}
Given that the effective action associated to a quantum \(L_\infty\)-superalgebra \(\mathfrak g\) defines an \(L_\infty\)-algebra \(\mathfrak g_{\mathrm{1PI}}\), we may consider the set of its Maurer--Cartan elements; these then constitute the set of quantum-corrected on-shell backgrounds.
Note that, using \eqref{eq:classical_MC_follows_from_quantum_MC}, a quantum Maurer--Cartan element of the quantum minimal model \(\operatorname H(\mathfrak g)\) is automatically a (non-quantum) Maurer--Cartan element of \(\operatorname H(\mathfrak g)\) regarded as a (non-quantum) \(L_\infty\)-superalgebra and, therefore, corresponds to a Maurer--Cartan element of \(\mathfrak g_{\mathrm{1PI}}\) up to equivalence.
For brevity, let us refer to a Maurer--Cartan element of \(\mathfrak g_{\mathrm{1PI}}\) as a \emph{one-particle-irreducible Maurer--Cartan element}.
Since \(\mathfrak g_{\mathrm{1PI}}\) is obtained from \(\mathfrak g\) by \(\mathcal O(\hbar)\) corrections,
the \(\mathcal O(\hbar^0)\) component of a one-particle-irreducible Maurer--Cartan element is a (non-quantum) Maurer--Cartan element of the (non-quantum) \(L_\infty\)-algebra \((\mathfrak g,\mu_i^0)\) obtained by forgetting all higher-genus operations \(\{\mu_i^g\}_{g>0}\). (In physics terms, this means that a stationary point of the quantum effective action is perturbatively corrected from a stationary point of the bare classical action by \(\mathcal O(\hbar)\) terms.)

\begin{example}[1PI action of scalar $\phi^4$.] Consider an \(L_\infty\)-algebra \(\mathfrak g\) corresponding to the classical scalar field theory
\begin{equation}
    S = \int\mathrm d^4\,\frac12\phi(\square-m^2)\phi-\frac1{4!}\lambda\phi^4
\end{equation}
The quantum effective action to one loop order is then\footnote{dropping a cosmological constant term, which does not enter into the (quantum) \(L_\infty\)-algebra}  of the form \cite[(16.2.15)]{Weinberg:1996kr}
\begin{equation}
\begin{split}
    \Gamma &= \int\mathrm d^4x\,\frac12\phi(\square-m_{\mathrm R}^2)\phi-\frac1{4!}\lambda_{\mathrm R}\phi^4\\
    &\qquad-\frac\hbar {64\piup^2}\left(m_{\mathrm R}^2+\frac12\lambda_{\mathrm R}\phi^2\right)^2\ln\left(m_{\mathrm R}^2+\frac12\lambda_{\mathrm R}\phi^2\right)+\dotsb \\
    &= S +   \mathcal O(\hbar),\label{eq:effective-action}
\end{split}
\end{equation}
where the ellipses include   neglected  derivative interaction terms and
\begin{align}
    \lambda_{\mathrm R}&=\lambda+\mathcal O(\hbar)&
    m_{\mathrm R}&=m+\mathcal O(\hbar)
\end{align}
are the renormalised coupling constants. A Maurer--Cartan element \(\phi\) of \(\mathfrak g\) makes \(S\) stationary, and a one-particle-irreducible Maurer--Cartan element \(\psi=\phi+\mathcal O(\hbar)\) makes \(W=S+\mathcal O(\hbar)\) stationary.
\end{example}
\section{Twisting qua classical backgrounds}\label{sec:classical_background}
Given a classical field theory defined by an action principle \(S[\phi]\) whose space of fields is a linear space, one can construct the theory with respect to a classical background \(\phi_0\) by
\begin{equation}S_{\phi_0}[\tilde\phi]=S[\tilde\phi+\phi_0].\end{equation}
For arbitrary \(\phi_0\), the resulting action \(S_{\phi_0}[\tilde\phi]\) will in general contain tadpole terms (i.e.~terms linear in \(\tilde\phi\)), signalling the nonstationarity of the putative `vacuum' \(\tilde\phi=0\) (i.e.~\(\phi=\phi_0\)). However, when \(\phi_0\) is a solution to the classical equations of motion, then \(S_{\phi_0}\) will not have tadpole terms, and the putative `vacuum' \(\tilde\phi=0\) is stationary (although not necessarily stable).
For a quantum field theory, the same applies except that \(\phi\) must instead be stationary with respect to the one-particle irreducible effective action in order for the effective action with background to not have tadpole terms.

The above well known physics lore
has a natural interpretation in terms of \(L_\infty\)-algebras. Given a perturbative classical field theory, we can formulate it as a cyclic \(L_\infty\)-algebra \(\mathfrak g\) such that the minimal model yields the scattering amplitudes \cite{Macrelli:2019afx} or CFT correlators \cite{Chiaffrino:2023wxk,Alfonsi:2024utl}. (For reviews of the \(L_\infty\)-algebraic formalism to perturbative quantum field theory, see \cite{Jurco:2018sby,Jurco:2020yyu,Borsten:2024gox}.)
Now, a Maurer--Cartan element for \(\mathfrak g\) is precisely a field (rather than an antifield, Fadeev--Popov ghost, etc.)
that satisfies the equations of motion. The twist corresponds to turning on a background field to obtain a new \(L_\infty\)-algebra, which corresponds to the theory atop the classical backgrounds. More generally, we can twist by any field that need not fulfil the Maurer--Cartan equation \eqref{eq:Maurer-Cartan}, in which case we obtain a curved \(L_\infty\)-algebra, which corresponds to a field theory action containing tadpole terms.

In the quantum case, a perturbative quantum field theory may be formulated as a cyclic quantum \(L_\infty\)-algebra such that the minimal model yields the loop-level scattering amplitudes \cite{Jurco:2019yfd}.
Then a one-particle-irreducible Maurer--Cartan element corresponds to a field configuration stationary with respect to the one-particle-irreducible effective action,
and twisting by such a field corresponds to turning on this background.

\begin{example}[The Mexican-hat potential]
Let us consider a scalar field theory with a Mexican-hat potential:
\begin{equation}
  \label{eq:mexicanscalar}
    S[\phi] = \int\mathrm d^dx\,\left(\frac12\phi\square\phi+\frac12\mu^2\phi^2-\frac1{4!}\lambda\phi^4\right).
\end{equation}
Perturbation theory about \(\phi=0\) shows that this `vacuum' is in fact unstable since the spectrum contains tachyons as seen by the wrong positive sign for the mass term \(\frac12\mu^2\phi^2\).

The above correspondence between perturbative field theories and (cyclic) homotopy Lie algebras allows us to recast this in terms of $L_\infty$-algebras. The \(L_\infty\)-algebra $\fg^{\text{MH}}$ corresponding to \eqref{eq:mexicanscalar} is
\begin{equation}
    \mathfrak g = (0\to\mathcal C^\infty(\mathbb R^d)[-1]\xrightarrow{-\square-\mu^2}\mathcal C^\infty(\mathbb R^d)[-2]\to0),
\end{equation}
with
\begin{subequations}
\begin{align}
    \mu_1(\phi[-1])&=-(\square+\mu^2)\phi[-2]\\
    \mu_3(\phi_1[-1],\phi_2[-1],\phi_3[-1])&=(\lambda\phi_1\phi_2\phi_3)[-2]
\end{align}
\end{subequations}
for fields \(\phi[-1],\phi_1[-1],\phi_2[-1],\phi_3[-1]\in\mathcal C^\infty(\mathbb R^d)[-1]\)
with all other \(\mu_i\) vanishing.

A Maurer--Cartan element of $\fg^{\text{MH}}$ is a degree-one element (that is, a field rather than an anti-field) which is a solution to the equations of motion:
\begin{equation}
    (\square+\mu^2)\phi-\frac1{3!}\lambda\phi^3=0.
\end{equation}
In particular, a constant value
\begin{equation}
  \label{eq:MHMC}Q=\mu\sqrt{6/\lambda}\end{equation}
is a solution to the equations of motion. The twist of the $L_\infty$-algebra $\fg^{\text{MH}}$ with respect to the Maurer--Cartan element \eqref{eq:MHMC} is the \(L_\infty\)-algebra
\begin{equation}
    \fg_Q^{\text{MH}} = (0\to\mathcal C^\infty(\mathbb R^d)[-1]\xrightarrow{-(\square-2\mu^2)}\mathcal C^\infty(\mathbb R^d)[-2]\to0),
\end{equation}
with
\begin{subequations}
\begin{align}
\mu_1^Q(\phi[-1])&=-(\square-2\mu^2)\phi[-2]\\
\mu_2^Q(\phi_1[-1],\phi_2[-1])&=\mu\sqrt{6\lambda}\phi_1\phi_2[-2]\\
\mu_3^Q(\phi_1[-1],\phi_2[-1],\phi_3[-1])&=\lambda\phi_1\phi_2\phi_3[-2]
\end{align}
\end{subequations}
which corresponds to the action
\begin{equation}
    S_Q[\tilde\phi] = S[Q+\tilde\phi] = \int\mathrm d^dx\,\left(\frac12\tilde\phi(\square-2\mu^2)\tilde\phi-\frac{\mu\sqrt{6\lambda}}{3!}\tilde\phi^3-\frac1{4!}\lambda\tilde\phi^4\right)
\end{equation}
corresponding to perturbation theory about the true vacuum, since now the mass term \(-\mu^2\phi^2\) has the correct negative sign.

We can, furthermore, consider quantum corrections. Consider \eqref{eq:effective-action} but with the wrong-sign mass term:
\begin{align}
W &=  \int\mathrm d^4x\,\frac12\phi(\square+\mu_{\mathrm R}^2)\phi-\frac1{4!}\lambda_{\mathrm R}\phi^4
    -\frac\hbar {64\piup^2}M^2(\phi)\ln M^2(\phi)+\dotsb \notag\\
&= S + \mathcal O(\hbar),
\end{align}
where
\begin{align}
    \mu_{\mathrm R}&=\mu+\mathcal O(\hbar) &
    \lambda_{\mathrm R}&=\lambda+\mathcal O(\hbar) \\
    M^2(\phi)&=-\mu_{\mathrm R}^2+\frac12\lambda_{\mathrm R}\phi^2
\end{align}
are the renormalised quantities. Then a one-particle-irreducible Maurer--Cartan element \(Q\) is a stationary point of the quantum effective action \(W\). One choice is the constant value
\begin{equation}
    Q = \mu_{\mathrm R}\sqrt{6/\lambda_{\mathrm R}}+\mathcal O(\hbar).
\end{equation}
Twisting by this, we obtain the quantum \(L_\infty\)-algebra whose effective action is
\begin{multline}
    \Gamma_Q[\tilde\phi] = \Gamma[Q+\tilde\phi] = \int\mathrm d^dx\,\Big(\frac12\tilde\phi(\square-2\mu_{\mathrm R}^2)\tilde\phi-\frac{\mu_{\mathrm R}\sqrt{6\lambda_{\mathrm R}}}{3!}\tilde\phi^3-\frac1{4!}\lambda_{\mathrm R}\tilde\phi^4\\
   + M^2(Q+\tilde\phi)\ln M^2(Q+\tilde\phi)+\dotsb\Big).
\end{multline}
\end{example}

\section{Classical backgrounds and anomalies}\label{sec:anomaly}
In the previous section, we introduced a classical background for a field that was already present in the \(L_\infty\)-algebra.
In this section, we instead enlarge the \(L_\infty\)-algebra first by coupling it to nondynamical fields and then twist to put the theory on nontrivial backgrounds.
If we apply this procedure to one-particle-irreducible effective actions, we may observe the non-conservation of classically would-be-conserved currents that signal the presence of a quantum anomaly.
(For reviews of anomalies, see e.g.~\cite{Shifman:1988zk,Bertlmann:1996xk,Harvey:2005it}. For a different perspective on anomalies in the Batalin--Vilkovisky formalism, see \cite{Howe:1990pz}.)

\subsection{Background gauge fields and gauge anomalies}
We start with a field theory that has a classical global symmetry \(G\) that may be anomalous, described by a quantum \(L_\infty\)-superalgebra \(\mathfrak g\).
Now, we may couple this theory to a nondynamical \(G\)-gauge field \(A\), but with neither a kinetic term for \(A\) nor ghosts for the gauge transformations, to produce a larger quantum \(L_\infty\)-superalgebra \(\tilde{\mathfrak g}\).
(This procedure is possible even if \(G\) is anomalous, since there are no ghosts or gauge transformations,
but not always unique since there are always curvature ambiguities (improvement terms).)
Since \(A\) has no dynamics, it is not constrained by any equations of motion.

We may then twist the enlarged algebra \(\tilde{\mathfrak g}\) by a background field \(A_0\) to obtain \(\tilde{\mathfrak g}_{A_0}\). This will in general not be a one-particle-irreducible Maurer--Cartan element since the corresponding one-particle-irreducible effective theory \((\tilde{\mathfrak g}_{A_0})_\mathrm{1PI}\) will have a tadpole term (corresponding to nullary operations \(\mu_0=\sum_g\hbar^g\mu_0^g\)), which represents the current \(j\) induced by the background field \(A_0\); one detects an anomaly when \(\partial\cdot j\ne0\).
The same discussion applies to abelian \(p\)-form symmetries except that the corresponding connection \(A\) is then a \((p+1)\)-form potential \cite{Borsten:2024gox}.

For example, in \(d=2n\) spacetime dimensions, consider a Dirac field \(\Psi\) with action
\begin{equation}\label{eq:orig_action_gauge}
    S[\phi] \coloneqq \int\mathrm d^dx\,\bar\Psi\partial_\mu\gamma^\mu\Psi
\end{equation}
corresponding to the \(L_\infty\)-superalgebra
\begin{equation}
    \mathfrak g\coloneqq\left(0\to\underbrace{\Omega^0(M;\Piup E)[-1]}_\Psi\to\underbrace{\Omega^d(M;\Piup E^*)[d-2]}_{\Psi^+}\to0\right),
\end{equation}
where \(\Piup E\) is the Dirac spinor bundle with the \(\mathbb Z_2\) parity reversed.
This has the global \(\operatorname U(1)\times\operatorname U(1)\) vector and axial symmetries
\begin{equation}
    \Psi\mapsto\exp(\mathrm i\alpha+\mathrm i\beta\gamma^{d+1})\Psi
\end{equation}
for any real numbers \(\alpha,\beta\in\mathbb R\).
It is well known that if we gauge one of the two \(\operatorname U(1)\) symmetries, the other becomes anomalous and hence cannot be consistently gauged --- if we try to do so, longitudinal modes of the other gauge field will not decouple, and we lose unitarity.
However, we may `nearly' gauge both
in the sense of introducing non-dynamical gauge fields \(A^\mathrm{vect}\) and \(A^\mathrm{axial}\) (without kinetic terms or gauge transformations) and changing the derivative to the covariant derivative:
\begin{equation}\label{eq:gauged_action}
    S_\text{BV} \coloneqq \int\mathrm d^dx\,\bar\Psi\gamma^\mu(\partial_\mu+\mathrm iA^{\mathrm{vect}}_\mu+\mathrm iA^{\mathrm{axial}}_\mu\gamma^{d+1})\Psi,
\end{equation}
corresponding to the \(L_\infty\)-superalgebra
\begin{multline}
    \tilde{\mathfrak g}\coloneqq\Bigg(0\to\underbrace{\Omega^1(M)}_{A^{\mathrm{vect}}}\oplus\underbrace{\Omega^1(M)}_{A^{\mathrm{axial}}}\oplus\underbrace{\Omega^0(M;\Piup E)[-1]}_\Psi\\\to\underbrace{\Omega^{d-1}(M)[d-3]}_{A^{\mathrm{vect}+}}\oplus\underbrace{\Omega^{d-1}(M)[d-3]}_{A^{\mathrm{axial}+}}\oplus\underbrace{\Omega^d(M;\Piup E)[d-2]}_{\Psi^+}\to0\Bigg).
\end{multline}
Note that we do \emph{not} include ghosts for any would-be gauge transformations or kinetic terms for the gauge fields, which remain nondynamical.

This now admits a classical local \(\operatorname U(1)\times\operatorname U(1)\) symmetry
\begin{align}
    \Psi&\mapsto\exp(\mathrm i\alpha(x)+\mathrm i\beta(x)\gamma^{d+1})\Psi,&
    A^{\mathrm{vect}}_\mu&\mapsto A_\mu-\partial_\mu\alpha,&
    A^{\mathrm{axial}}_\mu&\mapsto A_\mu-\partial_\mu\beta,
\end{align}
which is broken at the quantum level.
If one sets the background field to be \(A^{\mathrm{vect}}_\mu=A^{\mathrm{axial}}_\mu=0\), one recovers the original action \eqref{eq:orig_action_gauge} as a special case of \eqref{eq:gauged_action}.

Now, we may twist \(\tilde{\mathfrak g}\) by background fields \(A_0\coloneqq(A^{\mathrm{vect}}_0,A^{\mathrm{axial}}_0)\) to obtain the curved \(L_\infty\)-superalgebra \(\tilde{\mathfrak g}_{A_0}\).
If we then pass to the one-particle-irreducible effective action represented by the \emph{curved} \(L_\infty\)-superalgebra \((\tilde{\mathfrak g}_{A_0})_{\mathrm{1PI}}\),
this in general has a tadpole term, corresponding to (anomalously) non-conserved currents.
For simplicity, let us suppose that \(A^{\mathrm{axial}}_0=0\), that is, we only turn on a background vector gauge field.
In that case, there are tadpole terms in the one-particle-irreducible effective action generated by one-point one-loop diagrams such as
\begin{equation}
    \feyn{\otimes!g{A^{\mathrm{vect}}}f0flAfluVf0!g{A^{\mathrm{axial}}}}\feynstrut01\strut\quad\text{(if \(d=2\))}\qquad\text{or}\qquad
    \Diagram{\negmedspace\negmedspace\negmedspace\smash\otimes\\!{gd}{A^{\mathrm{vect}}}\\\scalebox{1.414}{\(\feyn{fs0}\)}\medspace\rotatebox{90}{\(\feyn a\)}\negmedspace\scalebox{1}[0.7071]{\(\feyn{fv}\)}fdA!g{A^{\mathrm{axial}}}\\\scalebox{1.414}{\(\feyn{fs0}\)}\scalebox{1}[0.7071]{\(\feyn{fv}\)}fuV\\!{gu}{A^{\mathrm{vect}}}\\\negmedspace\negmedspace\negmedspace\smash\otimes}
    \quad\text{(if \(d=4\))},
\end{equation}
where \(\otimes\) refers to contraction with the background \(A^{\mathrm{vect}}_0\).
These tadpoles \(\mu_0^{(\tilde{\mathfrak g}_{A_0})_{\mathrm{1PI}}}\) correspond to the axial current, i.e.
\begin{equation}
    \langle A^{\mathrm{axial}}_\mu,\mu_0^{(\tilde{\mathfrak g}_{A_0})_{\mathrm{1PI}}}\rangle=j^\mathrm{axial}_\mu,
\end{equation}
induced by the background gauge field \(A_0^{\mathrm{vect}}\).

In this example, the Adler--Bell--Jackiw chiral anomaly may then be detected by the non-vanishing of the expression
\begin{equation}\label{eq:ABJ_anomaly}
    \partial_\mu\langle j^\mu_\mathrm{axial}\rangle=
    \partial_\mu \langle A^{\mathrm{axial},\mu},\mu_0^{(\tilde{\mathfrak g}_{A_0})_{\mathrm{1PI}}}\rangle
    \propto \star(\overbrace{F^{\mathrm{vect}}_0\wedge\dotso\wedge F^{\mathrm{vect}}_0}^{d/2}),
\end{equation}
where \(F^{\mathrm{vect}}_0=\mathrm dA^{\mathrm{vect}}_0\) is the background field strength (and, as usual, we insist that the vector current is anomaly free).

\subsection{Curved background spacetime and mixed anomalies}
Similar considerations apply for mixed anomalies except that we must introduce a background metric \(g\) on the spacetime manifold \(M\) in addition to a background gauge field \(A\).
Given a theory with a global symmetry given by a Lie group \(H\) (with Lie algebra \(\mathfrak h\)) that is described by a quantum \(L_\infty\)-superalgebra \(\mathfrak g\),
we may then couple it to a non-dynamical metric \(g\) and a background gauge field \(A\) (and their antifields \(g^+\) and \(A^+\)) to produce the larger quantum \(L_\infty\)-superalgebra
\begin{multline}
    \tilde{\mathfrak g}=\mathfrak g\oplus \Bigg(
    0\to\overbrace{\Omega^0(M;\mathrm T^{*\odot2}M)[-1]}^g\oplus\overbrace{\Omega^1(M;\mathfrak h)}^A\\
    \to
    \underbrace{\Omega^d(M;\mathrm T^{\odot2}M)[d-2]}_{g^+}\oplus\underbrace{\Omega^{d-1}(M;\mathfrak h)[d-3]}_{A^+}
    \to0
    \Bigg),
\end{multline}
where the antifields are tensor densities, or equivalently (after picking an orientation of spacetime) top-degree differential forms valued in powers of the (co)tangent bundle.

As before, we may twist \(\tilde{\mathfrak g}\) by background fields \((g_0,A_0)\) to obtain the curved \(L_\infty\)-superalgebra \(\tilde{\mathfrak g}_{(g_0,A_0)}\).
If we then pass to the one-particle-irreducible effective action represented by the curved \(L_\infty\)-superalgebra \((\tilde{\mathfrak g}_{(g_0,A_0)})_{\mathrm{1PI}}\),
mixed gauge--gravitational anomalies may be detected by the non-vanishing of the expression
\begin{equation}
    \nabla_\mu\langle j^\mu_a\rangle=\partial_\mu\frac{\delta\mu_0^{(\tilde{\mathfrak g}_{(g_0,A_0)})_{\mathrm{1PI}}}}{\delta A_{0\mu}^a},
\end{equation}
where \(a\) is an adjoint index for \(\mathfrak h\).

\section{Twisting qua twisting}\label{sec:twistingtwisting}
Since the seminal work of Witten \cite{Witten:1988ze}, it is well known that many supersymmetric field theories admit twists by nilquadratic supersymmetry generators, which may be topological, holomorphic, or somewhere in between \cite{Elliott:2020ecf}. Such twisting is described by the twisting (in the sense of \cref{def:classicaltwist}) of the super-Poincaré Lie superalgebra \cite{Saberi:2021weg} encoding the supersymmetries of the theory. This section reviews
how this notion dovetails with the \(L_\infty\)-algebra formalism for scattering amplitudes.
\subsection{General construction and off/on-shell supersymmetry}
For many supersymmetric models, one finds that the supersymmetry algebra is only realised on shell. That is, the algebra closes only modulo the equations of motion. Oftentimes the failure of the supersymmetry algebra to close can be explained through the language of homotopy algebra; the would-be Lie algebra morphism $\rho$ from the supersymmetry algebra to the endomorphisms of the model is in fact an $L_\infty$-algebra morphism and there are higher maps, correcting the failure of $\rho$ to close off shell \cite{Eager:2021wpi}.

To set the stage, we are interested in theories which have a (not necessarily strict) symmetry given by a super Lie algebra $\fp=\fp^+\oplus \fp^-$. (That is, a strict flat $L_\infty$-superalgebra concentrated in cohomological degree $0$.) Usually these supersymmetry algebras are of the form:
\begin{equation}
  \label{eq:supersymmetry}
  \fp=\overline{\fder(\ft)} \ltimes \ft,
\end{equation}
where $\ft=\ft^+\oplus \ft^-$ is a Lie superalgebra whose only nontrivial bracket is defined by a map $\ft^-\otimes \ft^-\to \ft^+$, and $\overline{\fder(\ft)}$ is a subalgebra of the (bidegree-preserving) Lie algebra of derivations of $\ft$.
\begin{remark}
  Physically, $\ft^+$ is to be thought of as the translations of the underlying space of the theory while $\ft^-$ are the supersymmetries and $\fder(\ft)$ are the Lorentz  and $R$-symmetries of the theory. Lie superalgebras of this form encode supersymmetry algebras in any dimension with any amount of supersymmetry.
\end{remark}

\begin{remark}\label{rem:lift}
  Algebras of this form furthermore allow for a lift of the $\ZZ\times\ZZ_2$-grading to a $\ZZ\times \ZZ$-grading by putting $\overline{\fder(\ft)}$, $\ft^-$, and $\ft^+$ in bidegrees $(0,0)$ $(0,1)$, and $(0,2)$, respectively. This lift allows for a more tractable treatment of the algebraic structures and homotopy transfer.
\end{remark}

\bigskip
Let \(\mathfrak g\) be a cyclic \(L_\infty\)-superalgebra corresponding to a perturbative field theory enjoying a (not necessarily strict) global supersymmetry given by a Lie superalgebra \(\mathfrak p\) of the form \eqref{eq:supersymmetry}\footnote{The construction goes through for any Lie superalgebra $\fp$, but we will focus on this case.},  regarded as an \(L_\infty\)-superalgebra concentrated in the $(\ZZ\times\ZZ_2)$-bidegrees \((0,0)\) and \((0,1)\).

Since the global symmetry acts on the fields, there must be an action of \(\mathfrak p\) on \(\mathfrak g\). In the language of homotopy algebras there then exists a morphism of \(L_\infty\)-superalgebras
\begin{equation}
  \label{eq:algebraaction}
    \mathfrak p\rightsquigarrow\big(\Coder(\CEc(\mathfrak g))\big)^1
\end{equation}
where \(\CEc(\mathfrak g)\) is the Chevalley--Eilenberg cosuperalgebra of \(\mathfrak g\), and
\(\Coder(X)\) is the differential graded Lie superalgebra of $(\ZZ\times\ZZ_2)$-graded coderivations on a differential graded cosuperalgebra \(X\) (whose elements need not be chain maps), and whose differential is given by the commutator with the differential on \(X\), and the notation \((-)^1\) denotes the degree-one part \cite{Kajiura:2004xu,Kajiura:2005sn,Kajiura:2006mt}.

This is equivalent to a family of morphisms of degree $(2-p-q,0)$ for $p\geq1$ and $q\geq0$
\begin{equation}
  \label{eq:componentalgebraaction}
\rho^{(p,q)}\colon\fp^{\wedge p}\otimes \fg^{\wedge q}\to \fg,  
\end{equation}
with $\rho^{(0,k)}=\mu_{k}^\fg$,
satisfying some compatibility relations with the $L_\infty$-algebra structure on $\fp$ and $\fg$ \cite{Kajiura:2004xu,Kajiura:2005sn,Kajiura:2006mt} of the form:
\begin{align}
&0=\sum_{p+r=n}\sum_{\substack{\sigma\in\Sym(n)}}\frac{\pm1}{p!r!}
\rho^{(1+r,m)}(\mu^\fp_p(x_{\sigma(1)},\cdots,x_{\sigma(p)}),
x_{\sigma(p+1)}\cdots,x_{\sigma(n)};\phi_1,\cdots,\phi_m) \notag\\
&\quad +
  \sum_{\substack{\sigma,\tilde\sigma\in\Sym(n)\\p+r=n\\i+s+j=m}}
\frac{\pm1}{p!r!i!s!j!}
\rho^{(p,i+1+j)}\bigg(x_{\sigma(1)},\dotsc,x_{\sigma(p)};\phi_{\tilde{\sigma}(1)},\dotsc,\phi_{\tilde{\sigma}(i)},\label{AovL}\\[-1.8em]
&\qquad\qquad\qquad\quad\!\!\rho^{(r,s)}(x_{\sigma(p+1)},\dotsc,x_{\sigma(n)};\phi_{\tilde{\sigma}(i+1)},\dotsc,\phi_{\tilde{\sigma}(i+s)}),
\phi_{\tilde{\sigma}(i+s+1)},\dotsc,\phi_{\tilde{\sigma}(m)}\bigg).\notag 
\end{align}

\begin{remark}
  Algebraic structures of this form are essentially a symmetrised version of an $A_\infty$-algebra over an $L_\infty$-algebra (or open--closed homotopy algebra), introduced by Kajiura and Stasheff in \cite{Kajiura:2004xu,Kajiura:2005sn,Kajiura:2006mt}, where an $L_\infty$-algebra acts on an $A_\infty$-algebra through homotopy derivations
\cite{PMIHES_1977__47__269_0,TolleyThesis,LadaTolley,2014arXiv1409.1691D}. See also \cite{Jonsson:2024rie} for open--closed homotopy algebras in supersymmetry.  
\end{remark}

Given the two $L_\infty$-superalgebras $\fp$ and $\fg$ and the maps $\rho^{(p,q)}$, we may form an \(L_\infty\)-superalgebra whose Chevalley--Eilenberg coalgebra is of the following form.
There is an isomorphism of differential graded cosuperalgebras
\begin{equation}
  (\CEc(\fp),d_{\CE}^\fg)\otimes(\CEc(\fg),d_{\CE}^\fg)\cong (\CEc(\fp\oplus\fg),d_{\CE}^{\fp}+d_{\CE}^{\fg}).
\end{equation}
We extend the $L_\infty$-map \eqref{eq:algebraaction} to a degree $(1,0)$-coderivation $d_{\CE}^\rho$on $\CEc(\fp\oplus\fg)$. Then, $d_{\CE}^\rho$ can be added to form the full algebra

\begin{equation}
  \label{eq:CEsemidirect}
 (\CEc(\fp\oplus\fg),d^{\fp\ltimes\fg}_{\CE}\coloneq d_{\CE}^{\fp}+d_{\CE}^{\fg}+d_{\CE}^\rho)
\end{equation}
The nilquadraticity of $d^{\fp\ltimes\fg}_{\CE}$ is then equivalent to the homotopy Jacobi identities on $\fp$ and $\fg$, and the compatibility \eqref{AovL}. The proof of this is straightforward and similar to the ones in \cite{Kajiura:2004xu,Kajiura:2005sn,Kajiura:2006mt}. The resulting $L_\infty$-superalgebra is then 
 
\begin{equation}
  \label{eq:semidirect}
    \mathfrak p\ltimes\mathfrak g \coloneqq\Big(\fp\oplus\fg,(\{\mu_i^\fp,\rho^{(p,q)},\mu_j^\fg\})_{\substack{i,j\ge0\\p,q\ge1}}\Big).
  \end{equation}
    (This may be formally thought of as `gauging' the global super-Poincaré symmetry \(\mathfrak p\) since \(\mathfrak p\) sits in degree zero, which normally corresponds to Fadeev--Popov ghosts   \cite{Jurco:2018sby}.)

We would like to twist by a nilquadratic element of \(\mathfrak p\subset\mathfrak p\ltimes\mathfrak g\). Naïvely, this cannot correspond to a Maurer--Cartan element of \(\mathfrak p\ltimes\mathfrak g\) since \(\mathfrak p\) is concentrated in degree zero rather than one. For this purpose, following \cite{Saberi:2021weg}, we introduce a formal parameter \(u\) of degree \((1,-1)\) and take
\begin{equation}
    \mathfrak G[u]\coloneqq(\mathfrak p\ltimes\mathfrak g)\otimes {\mathbb C}[u].
\end{equation}
(Complexification is convenient and often necessary for twisting \cite{Elliott:2020ecf}.)
Then \(\mathfrak p[u]\subset \fG[u] \) contains a ready supply of bidegree \((1,0)\) elements of the form \(uQ\) where \(Q\in\mathfrak p^-\) satisfies \([Q,Q]_{\mathfrak p}=0\), and these form Maurer--Cartan elements inside \(\mathfrak G[u]\) with which we may twist.

More generally, we may twist by a Maurer--Cartan element of \(\mathfrak G[u]\) that is not purely restricted to \(\mathfrak p[u]\); this corresponds to a combination of taking a classical background and twisting, a possibility that will be fully exploited in \cref{sec:localisation}.

For supersymmetric models where supersymmetry is realised only on shell---that is, after imposing the equations of motion---the projective superspace (reviewed in \cite{Lindstrom:2007rq,Kuzenko:2010bd}), harmonic superspace (reviewed in \cite{Howe:1995md,Galperin:2001seg}) or the pure spinor (reviewed in \cite{Berkovits:2017ldz,Cederwall:2013vba,Cederwall:2022fwu,Eager:2021wpi}) formalisms often succeed in producing an off-shell representation at the cost of introducing infinite towers of auxiliary fields.

However, since such infinite towers are often unwieldy, it is sometimes convenient to be able to work directly with the on-shell supersymmetry representation.
For this, one can often lift the  \(\mathbb Z\times\mathbb Z_2\) bigrading to \(\mathbb Z\times\mathbb Z\) \cite{Eager:2021wpi,Jonsson:2024uyr} and apply a minimal-model construction to obtain an on-shell supersymmetry representation with finitely many components.

Supersymmetry algebras arising in physics usually admit such lifts (cf. \cref{rem:lift}).
Schematically, let $V$ be an $n$-dimensional (real) vector space, $\mathfrak{o}(V)$ the Lie algebra of orthogonal transformations of $V$, and $S$ a (linear combination of) spin representations of $\mathfrak o(V)$. Then, a (lifted) supersymmetry algebra (without $R$-symmetry) is of the form
\begin{equation}
    \mathfrak p = \mathfrak o(V)\ltimes(S[0,-1]\oplus V[0,-2] ).
\end{equation}
We then introduce a formal parameter \(u\) of bidegree \((1,-1)\):
\begin{equation}
    \mathfrak p[u] = \mathfrak o(V)\ltimes(S[0,-1]\oplus V[0,-2]) \otimes \mathbb C[u].
\end{equation}
Then \(\mathfrak p[u]\) contains a ready supply of degree \(1\) elements of the form \(uQ\) where \(Q\in S[0,-1]\) is such that $[Q,Q]=0$.

\subsection{Example with off-shell supersymmetry}
Here we provide a simple example with off-shell supersymmetry, that is, where the supersymmetry algebra is represented on the nose.
Let us consider a free theory containing a $\mathcal N=1$ chiral supermultiplet (sometimes called the Wess--Zumino multiplet) on Minkowski space \(\mathbb R^{1,3}\) (or, rather, complexified Minkowski space \(\mathbb C^4\)). The supersymmetry algebra is
\begin{equation}
    \mathfrak p = \mathfrak{o}(\mathbb C^4)\ltimes(\Piup S\oplus \mathbb C^4),
\end{equation}
where $S=S_+\oplus S_-$ is the Dirac spin representation, \(S_\pm\cong\mathbb C^2\) are the Weyl spinor representations, and $\Piup$ denotes shift in the $\ZZ_2$ degree. In addition to the brackets defined by the action of $ \mathfrak{o}(\mathbb C^4)$ on $\mathbb C^4$ and $S$, the only nontrivial brackets are defined by the isomorphism
\begin{equation}
  S_+\otimes S_-\cong \mathbb C^4.
\end{equation}

The \(L_\infty\)-superalgebra defining the supermultiplet is 
\begin{multline}
    \mathfrak g = (0\to\Omega^0(\mathbb C^4;\mathbb C\oplus\mathbb C\oplus\Piup S_+\oplus\overline{\mathbb C\oplus\mathbb C\oplus\Piup S_+})[-1]\xrightarrow{\mu_1}\\
    \Omega^4(\mathbb C^4;\mathbb C\oplus\mathbb C\oplus\Piup S_-\oplus\overline{\mathbb C\oplus\mathbb C\oplus\Piup S_-})[2]\to0),
\end{multline}
corresponding to the supermultiplet \((\phi,\psi,F)\)
with a complex scalar field \(\phi\), a Weyl fermion \(\psi\), and an auxiliary field \(F\),
and with the complex conjugate fields explicitly separated out due to the complexification (see e.g.~\cite[§10.1]{Elliott:2020ecf}).
The only nontrivial bracket is $\mu_1$, corresponding to the free equations of motion.

Gauging supersymmetry in the sense of adjoining global supersymmetry ghosts, we obtain the $L_\infty$-superalgebra 
\begin{equation}
    (\fp\ltimes \fg)[u] =(\mathfrak p\xrightarrow0\Omega^0(\mathbb C^4;\mathbb C\oplus\mathbb C\oplus S_+)[-1]\to\Omega^4(\mathbb C^4;\mathbb C\oplus\mathbb C\oplus S_-)[2]\to0)[u],
\end{equation}
where we have also adjoined a formal variable \(u\) of degree \((1,-1)\),
with differential $\mu_1$, and $\mu_2$ corresponding to the Lie algebra structure on $\fp$ and its action on $\fg$. 
  
A supercharge $Q\in \fp$ squares to zero precisely when $Q\in S_+$ or $Q\in S_-$. Thus, by picking such a $Q$ we can twist $\fp\ltimes \fg$, to obtain the twisted theory \cite{Saberi:2021weg}. Moreover, since $\mu_i\rvert_{ \fg\otimes\dotsb\otimes \fg}=0$ for \(i\ne1\),\footnote{This is because \(\mathfrak g\), describing a free theory, did not have a \(\mu_2\). Adjoining \(\mathfrak p\) produces new brackets between elements of \(\mathfrak p\) and elements of \(\mathfrak g\) but never amongst elements of \(\mathfrak g\).}
any element $(\phi,\bar\phi)\in\Omega^0(\mathbb C^4;\mathbb C\oplus\bar{\mathbb C})[-1]$ that satisfies the Klein--Gordon equation defines a Maurer--Cartan element.
Furthermore, if we pick \(Q\in S_+\), the holomorphic field \(\phi\) is \(Q\)-closed, so we may twist $(\fp\ltimes \fg)[u]$ by a linear combination
\begin{equation}
    \mathcal Q_\phi = \phi+uQ
\end{equation}
to obtain a twisted $L_\infty$-superalgebra with differential
\begin{equation}
    \mu_1^{(\fp\ltimes\fg)_{Q_\phi}}(-)=\mu_1(-)+\mu_2(\phi,-)+\mu_2(uQ,-).
\end{equation}

\subsection{Example with on-shell supersymmetry}
As a less trivial example, we consider the holomorphic twist of ten-dimensional \(\mathcal N=(1,0)\) supersymmetric Yang--Mills theory \cite{Baulieu:2010ch,Elliott:2020ecf} into holomorphic Chern--Simons theory in five complex-dimensions, which can be naturally formulated in the pure-spinor superfield formalism \cite{Saberi:2021weg}. The ten-dimensional \(\mathcal N=(1,0)\) super-Poincaré Lie superalgebra is (after lifting the $\ZZ_2$-grading to a $\ZZ$-grading) 
\begin{equation}
  \mathfrak p=\left(\mathfrak o\ltimes(V[0,-2]\oplus S[0,-1])\right)\otimes\mathbb C[u],
\end{equation}
where $V\cong\mathbb C^{10}$ and $S$ is one of the Weyl spin representations of $\mathfrak o(V)$.

Gauging supersymmetry in the sense of adjoining global ghosts to the \(L_\infty\)-superalgebra (and adjoining the formal variable $u$), we obtain the \(L_\infty\)-superalgebra
\begin{equation}
  \label{eq:10d}
    \left(\mathfrak p\ltimes(\mathfrak g\otimes X)\right)\otimes \mathbb C[u],
\end{equation}
where
\begin{equation}
    X = \mathbb C[x^\mu,\theta^\alpha,\lambda^\alpha]/(\lambda^\alpha\gamma^\mu_{\alpha\beta}\lambda^\beta)
\end{equation}
is a differential graded-commutative algebra encoding the pure spinor formulation of the ten-dimensional (colour-stripped) \(\mathcal N=(1,0)\) vector supermultiplet with
\begin{equation}
    \mathrm dx \coloneqq \lambda^\alpha\left(\frac\partial{\partial\theta^\alpha}-\gamma^\mu_{\alpha\beta}\theta^\beta\frac\partial{\partial x^\mu}\right)x,
\end{equation}
(and other differentials vanishing), and \(\mathfrak g\) is the colour Lie algebra (concentrated in bidegree \((0,0)\)).
The coordinates carry the bidegrees the bidegrees \cite[(3.15)~ff.]{Eager:2021wpi}\footnote{Strictly speaking, this only works if we regard spacetime as a supervariety (equipped with the structure sheaf of polynomial superfunctions) rather than a supermanifold (equipped with the structure sheaf of smooth superfunctions). This difficulty can be avoided by introducing a formal variable \(u\) that compensates for the nontrivial degrees \cite{Eager:2021wpi}.}
\begin{align}
|x|&=(0,-2)&|\theta|&=(0,-1)&|\lambda|&=(1,-1)&|u|&=(1,-1).
\end{align}
Now, a Maurer--Cartan element here can be either purely in \(\mathfrak p[u]\), or purely in \(\mathfrak g\otimes M[u]\), or in some combination. An example of the former is
\begin{equation}uQ\end{equation}
where \(Q\) is a pure spinor, i.e.~\(Q^\alpha\gamma_{\alpha\beta}^\mu Q^\beta=0\). An example of the latter is any bosonic solution to the classical equations of motion, such as
\begin{equation}
    a_\mu(\lambda\gamma^\mu\theta),
\end{equation}
where \(a_\mu\in\mathfrak g\times\mathbb C^{10}\) is a fixed Lie-algebra-valued 10-vector. This corresponds to a constant vacuum expectation value
of the gluon field \(A_\mu\) (which breaks part of Lorentz symmetry).
More generally, however, one can have a combination of both, as long as they are compatible.
In the above example, suppose that \(a_\mu(\lambda\gamma^\mu\theta)\)
is annihilated by \(Q\)\footnote{There are indeed nontrivial such elements, see \cite[§4.1.1]{Elliott:2020ecf} for details.}. In that case, we can twist by
\begin{equation}
  uQ+a_\mu(\lambda\gamma^\mu\theta),
\end{equation}
which corresponds to twisting by \(Q\) plus turning on a background field \(A\) that survives twisting. This produces a strict $L_\infty$-superalgebra (that is, $\mu_{i>2}=0$), whose underlying graded supervector space and $\mu_2$ coincides with that of \eqref{eq:10d}, and whose differential is of the form
\begin{equation}
  \mu_1^Q= \lambda\left(\frac\partial{\partial\theta}-\gamma^\mu\theta\frac\partial{\partial x^\mu}\right) + \mu_2\left(uQ,-\right)+\mu_2\left(a_\mu(\lambda\gamma^\mu\theta),-\right).
\end{equation}

In other words, twisting and classical backgrounds correspond to two extreme regions of a single moduli space of possible twisting backgrounds; equivalently, twisting corresponds to giving an expectation value of a certain global `ghost' (similar to twisted supergravity \cite{Costello:2016mgj} but `global').

\section{Supersymmetric  backgrounds and localisation}\label{sec:localisation}
In \cref{sec:twistingtwisting}, we twisted by an element corresponding to the symmetry algebra (or ghost); in \cref{sec:anomaly}, we coupled the theory to nondynamical background fields.
This section combines the two to show that the resulting notion of twisting reproduces the setup for localisation computations using supergravity backgrounds \cite{Festuccia:2011ws}, as reviewed in \cite{Dumitrescu:2016ltq}. From the \(L_\infty\)-algebraic formalism, it is clear that in fact we don't actually need the off-shell multiplet --- an on-shell multiplet suffices if we do homotopy transfer.
The observation that on-shell supersymmetry suffices for localisation was already made in \cite{Losev:2023gsq}. For more on the interaction of the Batalin--Vilkovisky formalism with localisation, see \cite{Cattaneo:2025wdw}. In this section we limit ourselves to a schematic discussion; detailed computations will be given in \cite{localisationpaper}.

For simplicity, let us consider the case of four-dimensional \(\mathcal N=1\) supersymmetry as in \cite{Festuccia:2011ws}. Recall that there are multiple possible supermultiplets containing the stress--energy tensor. Two of these are the Ferrara--Zumino multiplet, 
\begin{equation}(j_\mu,S_{\mu\alpha},x,T_{\mu\nu}),\end{equation}
which consists of  a nonconserved current, the supercharge, a complex scalar, and the stress--energy, respectively, and (in the cases where there is a conserved \(\operatorname U(1)\) R-symmetry current) the R-current supermultiplet, which is given by
\begin{equation}
    (j^{\text{(R)}}_\mu,S_{\mu\alpha},T_{\mu\nu},C_{\mu\nu}),
\end{equation}
corresponding to the conserved R-current, the supercharge, the stress--energy, and a conserved two-form current, respectively. These correspond to the old minimal supermultiplet \cite{Ferrara:1978em,Stelle:1978ye,Fradkin:1978jq} \((A_\mu,\psi_{\mu\alpha},M,g_{\mu\nu})\) and the new minimal supermultiplet \cite{Sohnius:1981tp} \((A^{\text{(R)}}_\mu,\psi_{\mu\alpha},g_{\mu\nu},B_{\mu\nu})\), respectively; gauging local diffeomorphisms and supersymmetries corresponds to introducing the fields of supergravity (for reviews, see \cite{Castellani:1991eu,Freedman:2012zz,Tanii:2014gaa,Ortin:2015hya}) but without introducing kinetic terms for the graviton and gravitino, or equivalently a suitable \(M_\text{Planck}\to\infty\) limit \cite{Festuccia:2011ws}.

For definiteness, let us consider localisation with an old-minimal supergravity background.
Consider chiral superfields taking values in a Kähler manifold \((N,\omega)\) with Kähler potential \(K\in\Omega^{0,0}(N)\) and superpotential \(W\in\Omega^{0,0}(N)\). After fake-gauging super-diffeomorphisms,
the \(L_\infty\)-algebra is concentrated in degrees \(0,1,2,3\) as
\begin{multline}
    \Omega^0(M,\mathrm TM\oplus S)\to\Omega^0(M,E)[-1]\\
    \to\Omega^d(M,E^*)[d-2]\to\Omega^d(M,\mathrm TM\oplus S)[d-3],
\end{multline}
where
\begin{equation}
    E = N\oplus\dotsb
    \oplus
    \mathrm T^{*\odot2}M
    \oplus
    \mathrm T^*M
    \oplus
    \mathbb C\oplus\mathbb C
\end{equation}
is the fibre bundle over \(M\) whose section is \((\phi,\psi,F,g_{\mu\nu},b_\mu,m,\bar m)\), and where the ghosts and ghost antifields correspond to diffeomorphisms and super-diffeomorphisms. Putting the theory on a curved manifold corresponds to twisting by a Maurer--Cartan element \((g,b,m,\bar m)\).

Now, for localisation, we adjoin a new formal coordinate \(u\) of degree \(+1\), so that we get
\begin{multline}
    \Omega^0(M,\mathrm TM\oplus S)[u]\to\Omega^0(M,E)[-1][u]\\\to\Omega^d(M,E^*)[d-2][u]\to\Omega^d(M,\mathrm TM\oplus S)[d-3][u].
\end{multline}
(The inner \([i]\) denotes suspension while the outer \([u]\) denotes a polynomial ring.)
If the background \((g,b,m,\bar m)\) that we twisted by is annihilated by a nilquadratic supersymmetry generator \(Q\in\Gamma(S)\),
then the sum
\begin{equation}
    uQ+(g,b,m,\bar m)
\end{equation}
is also a Maurer--Cartan element, and then localisation corresponds to twisting by this Maurer--Cartan element.

The \(L_\infty\)-algebraic formulation makes it clear that none of this discussion depends on having an off-shell strict realisation of supersymmetry --- an on-shell realisation that can be completed into a non-strict \(L_\infty\)-algebra representation of supersymmetry suffices for localisation \cite{Eager:2021wpi,Losev:2023gsq}.

\subsection{Old- versus new-minimal supergravity for localisation}
In the localisation literature, it is well known \cite{Festuccia:2011ws,Dumitrescu:2016ltq} that localisations using old-minimal versus new-minimal supergravities are not in general equivalent: for instance, the new-minimal localisation is only applicable to cases with an unbroken R-symmetry.
This is superficially in tension with our claim that in fact on-shell supersymmetry suffices in general for localisation, since old-minimal and new-minimal supergravities are but different off-shell formulations of the one and the same four-dimensional \(\mathcal N=1\) supergravity theory.

However, there are in general inequivalent ways to couple a given matter theory to supergravity, such that the resulting theories are different on shell.
For a simpler situation, consider coupling a matter theory with a stress--energy tensor \(T_{\mu\nu}\) to ordinary gravity:
\begin{equation}
    S = \frac1{2\kappa}\int\sqrt{|\det g|}\left(R-\frac12g^{\mu\nu}T_{\mu\nu}\right)+\dotsb.
\end{equation}
The choice of a stress--energy tensor is in general not unique (see e.g.\ \cite{Dumitrescu:2016ltq}), since one can always add improvement terms
\begin{equation}
    T'_{\mu\nu}=T_{\mu\nu}+(\partial_\mu\partial_\nu-\deltaup_{\mu\nu}\partial^2)O
\end{equation}
for a scalar operator \(O\). Then one can instead take
\begin{equation}
    S' = \frac1{2\kappa}\int\sqrt{|\det g|}\left(R-\frac12g^{\mu\nu}T'_{\mu\nu}\right)+\dotsb,
\end{equation}
such that \(S\) and \(S'\) are physically inequivalent theories with different scattering amplitudes.

A similar situation obtains in supergravity.
Suppose that one has a four-dimensional non-gravitational \(\mathcal N=1\) supersymmetric theory \(\mathcal T\) with a conserved R-symmetry, and suppose that we couple it to \(\mathcal N=1\) old-minimal and new-minimal supergravities (assuming both couplings are possible).
In that case, \(\mathcal T\) has two different supermultiplets --- the Ferrara--Zumino supermultiplet and the \(\mathcal R\)-supermultiplet (see e.g.\ \cite{Dumitrescu:2016ltq}) --- that each contain a stress--energy tensor, but the two stress--energy tensors \(T_{\mathrm{FZ}}\), \(T_{\mathcal R}\) are related by improvement terms and in general differ.
Coupling to old-minimal supergravity entails using the Ferrara--Zumino supermultiplet while coupling to new-minimal supergravity entails using the \(\mathcal R\)-supermultiplet.
Therefore, even when auxiliary fields have been integrated out (thus obtaining an action with only on-shell local supersymmetry), the theories \(\mathcal T+\text{old-minimal}\) and \(\mathcal T+\text{new-minimal}\) in general are physically inequivalent.
The inequivalence persists when one takes a rigid \(M_\mathrm{Pl}\to\infty\) limit.

The situation is even clearer for those theories for which one of the two supermultiplets (Ferrara--Zumino and \(\mathcal R\)) does not exist: in this case, one is forced to couple to either old-minimal or new-minimal supergravity depending on which supermultiplet exists, and this can be seen in principle even after auxiliary fields have been integrated out.

\section*{Acknowledgements}
H.K. was partly supported by the Leverhulme Research Project Grant \textsc{rpg}-2021-092. H.K. thanks Pietro Capuozzo\textsuperscript{\orcidlink{0000-0002-6486-9923}} and Charles Alastair Stephen Young\textsuperscript{\orcidlink{0000-0002-7490-1122}} for helpful conversations.

\newcommand\cyrillic[1]{\fontfamily{Domitian-TOsF}\selectfont \foreignlanguage{russian}{#1}}
\bibliographystyle{unsrturl}
\bibliography{biblio}

\begin{thebibliography}{10}

\bibitem{Witten:1988xj}
Edward Witten.
\newblock Topological sigma models.
\newblock {\em Communications in Mathematical Physics}, 118(3):411--449,
  September 1988.
\newblock \href {https://doi.org/10.1007/BF01466725}
  {\path{doi:10.1007/BF01466725}}.

\bibitem{Costello:2018zrm}
Kevin~Joseph Costello and Davide Silvano~Achille Gaiotto.
\newblock Twisted holography.
\newblock {\em Journal of High Energy Physics}, 2025(1):87, January 2025.
\newblock \href {https://arxiv.org/abs/1812.09257} {\path{arXiv:1812.09257}},
  \href {https://doi.org/10.1007/JHEP01(2025)087}
  {\path{doi:10.1007/JHEP01(2025)087}}.

\bibitem{Beem:2014kka}
Christopher Beem, Leonardo Rastelli, and Balt~C. van Rees.
\newblock $\mathcal{W} $ symmetry in six dimensions.
\newblock {\em Journal of High Energy Physics}, 2015(05):017, May 2015.
\newblock \href {https://arxiv.org/abs/1404.1079} {\path{arXiv:1404.1079}},
  \href {https://doi.org/10.1007/JHEP05(2015)017}
  {\path{doi:10.1007/JHEP05(2015)017}}.

\bibitem{pirsa_PIRSA:23070028}
Ingmar~Akira Saberi.
\newblock Research talk 13: Six- and eleven-dimensional theories via superspace
  torsion and {P}oisson brackets.
\newblock In {\em Strings 2023}, Waterloo, Ontario, Canada, July 2023.
  Perimeter Institute for Theoretical Physics.
\newblock \href {https://doi.org/10.48660/23070028}
  {\path{doi:10.48660/23070028}}.

\bibitem{Costello:2016mgj}
Kevin~Joseph Costello and Si~Li{~(\begin{CJK*}{UTF8}{bsmi}李思\end{CJK*})}.
\newblock Twisted supergravity and its quantization, June 2016.
\newblock \href {https://arxiv.org/abs/1606.00365} {\path{arXiv:1606.00365}}.

\bibitem{Raghavendran:2021qbh}
Surya Raghavendran, Ingmar~Akira Saberi, and Brian~R. Williams.
\newblock Twisted eleven-dimensional supergravity.
\newblock {\em Communications in Mathematical Physics}, 402(2):1103--1166,
  2023.
\newblock \href {https://arxiv.org/abs/2111.03049} {\path{arXiv:2111.03049}},
  \href {https://doi.org/10.1007/s00220-023-04745-2}
  {\path{doi:10.1007/s00220-023-04745-2}}.

\bibitem{Cushing:2022uyd}
Jay Cushing.
\newblock {\em Twisted Supergravity, and Family {D}onaldson Invariants}.
\newblock PhD thesis, Rutgers, The State University of New Jersey, New Jersey,
  United States of America, October 2022.
\newblock \href {https://doi.org/10.7282/t3-sbm1-wj69}
  {\path{doi:10.7282/t3-sbm1-wj69}}.

\bibitem{Hahner:2024xrh}
Fabian Hahner.
\newblock {\em The Pure Spinor Superfield Formalism and Twisted Supergravity.
  \foreignlanguage{german}{Der Pure-Spinor-Superfeld-Formalismus und getwistete
  Supergravitation}}.
\newblock PhD thesis, \foreignlanguage{german}{Ruprecht-Karls-Universität
  Heidelberg}, Heidelberg, Baden-Württemberg, Germany, January 2024.
\newblock \href {https://doi.org/10.11588/heidok.00034380}
  {\path{doi:10.11588/heidok.00034380}}.

\bibitem{Hahner:2023kts}
Fabian Hahner and Ingmar~Akira Saberi.
\newblock Eleven-dimensional supergravity as a {C}alabi--{Y}au twofold.
\newblock {\em \foreignlanguage{latin}{Selecta Mathematica}, New Series},
  31(2):38, April 2025.
\newblock \href {https://arxiv.org/abs/2304.12371} {\path{arXiv:2304.12371}},
  \href {https://doi.org/10.1007/s00029-025-01024-x}
  {\path{doi:10.1007/s00029-025-01024-x}}.

\bibitem{10.4310/jdg/1214437665}
Simon~Kirwan Donaldson.
\newblock An application of gauge theory to four-dimensional topology.
\newblock {\em Journal of Differential Geometry}, 18(2):279--315, June 1983.
\newblock \href {https://doi.org/10.4310/jdg/1214437665}
  {\path{doi:10.4310/jdg/1214437665}}.

\bibitem{Witten:1988ze}
Edward Witten.
\newblock Topological quantum field theory.
\newblock {\em Communications in Mathematical Physics}, 117(3):353--386,
  September 1988.
\newblock \href {https://doi.org/10.1007/BF01223371}
  {\path{doi:10.1007/BF01223371}}.

\bibitem{Witten:1994cg}
Edward Witten.
\newblock Monopoles and four-manifolds.
\newblock {\em Mathematical Research Letters}, 1(6):769--796, 1994.
\newblock \href {https://arxiv.org/abs/hep-th/9411102}
  {\path{arXiv:hep-th/9411102}}, \href
  {https://doi.org/10.4310/MRL.1994.v1.n6.a13}
  {\path{doi:10.4310/MRL.1994.v1.n6.a13}}.

\bibitem{Festuccia:2011ws}
Guido Nicola~Innocenzo Festuccia and Nathan Seiberg.
\newblock Rigid supersymmetric theories in curved superspace.
\newblock {\em Journal of High Energy Physics}, 2011(06):114, June 2011.
\newblock \href {https://arxiv.org/abs/1105.0689} {\path{arXiv:1105.0689}},
  \href {https://doi.org/10.1007/JHEP06(2011)114}
  {\path{doi:10.1007/JHEP06(2011)114}}.

\bibitem{Dumitrescu:2016ltq}
Thomas~T. Dumitrescu.
\newblock An introduction to supersymmetric field theories in curved space.
\newblock {\em Journal of Physics A}, 50(44):443005, November 2017.
\newblock \href {https://arxiv.org/abs/1608.02957} {\path{arXiv:1608.02957}},
  \href {https://doi.org/10.1088/1751-8121/aa62f5}
  {\path{doi:10.1088/1751-8121/aa62f5}}.

\bibitem{Elliott:2020ecf}
Christopher~J. Elliott, Pavel Safronov{~(\cyrillic{Павел
  Сафронов})}, and Brian~R. Williams.
\newblock A taxonomy of twists of supersymmetric {Y}ang--{M}ills theory.
\newblock {\em \foreignlanguage{latin}{Selecta Mathematica}, New Series},
  28(4):73, September 2022.
\newblock \href {https://arxiv.org/abs/2002.10517} {\path{arXiv:2002.10517}},
  \href {https://doi.org/10.1007/s00029-022-00786-y}
  {\path{doi:10.1007/s00029-022-00786-y}}.

\bibitem{Elliott:2020uwn}
Christopher~J. Elliott and Owen Gwilliam.
\newblock Spontaneous symmetry breaking: A view from derived geometry.
\newblock {\em Journal of Geometry and Physics}, 162:104096, April 2021.
\newblock \href {https://arxiv.org/abs/2008.03599} {\path{arXiv:2008.03599}},
  \href {https://doi.org/10.1016/j.geomphys.2020.104096}
  {\path{doi:10.1016/j.geomphys.2020.104096}}.

\bibitem{Jurco:2018sby}
Branislav Jurčo, Lorenzo Raspollini, Christian Sämann, and Martin Wolf.
\newblock \({L}_\infty\)-algebras of classical field theories and the
  {B}atalin--{V}ilkovisky formalism.
\newblock {\em \foreignlanguage{german}{Fortschritte der Physik}},
  67(7):1900025, July 2019.
\newblock \href {https://arxiv.org/abs/1809.09899} {\path{arXiv:1809.09899}},
  \href {https://doi.org/10.1002/prop.201900025}
  {\path{doi:10.1002/prop.201900025}}.

\bibitem{Arvanitakis:2019ald}
Alexandros~Spyridion Arvanitakis{~({\textgreekfont Αλέξανδρος
  Σπυριδίων Αρβανιτάκης})}.
\newblock The ${L}_\infty$-algebra of the {S}-matrix.
\newblock {\em Journal of High Energy Physics}, 2019(07):115, July 2019.
\newblock \href {https://arxiv.org/abs/1903.05643} {\path{arXiv:1903.05643}},
  \href {https://doi.org/10.1007/JHEP07(2019)115}
  {\path{doi:10.1007/JHEP07(2019)115}}.

\bibitem{Macrelli:2019afx}
Tommaso Macrelli, Christian Sämann, and Martin Wolf.
\newblock Scattering amplitude recursion relations in
  {B}atalin-{V}ilkovisky--quantizable theories.
\newblock {\em Physical Review D}, 100(4):045017, August 2019.
\newblock \href {https://arxiv.org/abs/1903.05713} {\path{arXiv:1903.05713}},
  \href {https://doi.org/10.1103/PhysRevD.100.045017}
  {\path{doi:10.1103/PhysRevD.100.045017}}.

\bibitem{Jurco:2019yfd}
Branislav Jur\v{c}o, Tommaso Macrelli, Christian S\"amann, and Martin Wolf.
\newblock Loop amplitudes and quantum homotopy algebras.
\newblock {\em Journal of High Energy Physics}, 2020(07):003, July 2020.
\newblock \href {https://arxiv.org/abs/1912.06695} {\path{arXiv:1912.06695}},
  \href {https://doi.org/10.1007/JHEP07(2020)003}
  {\path{doi:10.1007/JHEP07(2020)003}}.

\bibitem{Saemann:2020oyz}
Christian Sämann and Emmanouil Sfinarolakis~({\textgreekfont
  Εμμανουήλ Σφηναρολάκης}).
\newblock Symmetry factors of {F}eynman diagrams and the homological
  perturbation lemma.
\newblock {\em Journal of High Energy Physics}, 2020(12):088, December 2020.
\newblock \href {https://arxiv.org/abs/2009.12616} {\path{arXiv:2009.12616}},
  \href {https://doi.org/10.1007/JHEP12(2020)088}
  {\path{doi:10.1007/JHEP12(2020)088}}.

\bibitem{Jurco:2020yyu}
Branislav Jurčo, Hyungrok Kim~(\begin{CJK*}{UTF8}{bsmi}金炯錄\end{CJK*}),
  Tommaso Macrelli, Christian Sämann, and Martin Wolf.
\newblock Perturbative quantum field theory and homotopy algebras.
\newblock {\em Proceedings of Science}, 376:199, August 2020.
\newblock \href {https://arxiv.org/abs/2002.11168} {\path{arXiv:2002.11168}},
  \href {https://doi.org/10.22323/1.376.0199} {\path{doi:10.22323/1.376.0199}}.

\bibitem{Borsten:2021hua}
Leron Borsten, Hyungrok Kim~(\begin{CJK*}{UTF8}{bsmi}金炯錄\end{CJK*}),
  Branislav Jurčo, Tommaso Macrelli, Christian Sämann, and Martin Wolf.
\newblock Double copy from homotopy algebras.
\newblock {\em \foreignlanguage{german}{Fortschritte der Physik}},
  69(8--9):2100075, 2021.
\newblock \href {https://arxiv.org/abs/2102.11390} {\path{arXiv:2102.11390}},
  \href {https://doi.org/10.1002/prop.202100075}
  {\path{doi:10.1002/prop.202100075}}.

\bibitem{Borsten:2021gyl}
Leron Borsten, Hyungrok Kim~(\begin{CJK*}{UTF8}{bsmi}金炯錄\end{CJK*}),
  Branislav Jur\v{c}o, Tommaso Macrelli, Christian Sämann, and Martin Wolf.
\newblock Tree-level color\textendash{}kinematics duality implies loop-level
  color\textendash{}kinematics duality up to counterterms.
\newblock {\em Nuclear Physics B}, 989:116144, 2023.
\newblock \href {https://arxiv.org/abs/2108.03030} {\path{arXiv:2108.03030}},
  \href {https://doi.org/10.1016/j.nuclphysb.2023.116144}
  {\path{doi:10.1016/j.nuclphysb.2023.116144}}.

\bibitem{Borsten:2022ouu}
Leron Borsten, Hyungrok Kim~(\begin{CJK*}{UTF8}{bsmi}金炯錄\end{CJK*}),
  Branislav Jurčo, Tommaso Macrelli, Christian Sämann, and Martin Wolf.
\newblock Colour-kinematics duality, double copy, and homotopy algebras.
\newblock {\em Proceedings of Science}, 414:426, June 2023.
\newblock \href {https://arxiv.org/abs/2211.16405} {\path{arXiv:2211.16405}},
  \href {https://doi.org/10.22323/1.414.0426} {\path{doi:10.22323/1.414.0426}}.

\bibitem{1983PhRvD..28.2567B}
Igor~Anatolievich Batalin{~(\cyrillic{Игорь Анатольевич
  Баталин})} and Grigory~Alexandrovich
  Vilkovisky{~(\cyrillic{Григорий Александрович
  Вилковыский})}.
\newblock Quantization of gauge theories with linearly dependent generators.
\newblock {\em Physical Review D}, 28(10):2567--2582, November 1983.
\newblock \href {https://doi.org/10.1103/PhysRevD.28.2567}
  {\path{doi:10.1103/PhysRevD.28.2567}}.

\bibitem{Batalin:1985qj}
Igor~Anatolievich Batalin{~(\cyrillic{Игорь Анатольевич
  Баталин})} and Grigory~Alexandrovich
  Vilkovisky{~(\cyrillic{Григорий Александрович
  Вилковыский})}.
\newblock Existence theorem for gauge algebra.
\newblock {\em Journal of Mathematical Physics}, 26(1):172--184, January 1985.
\newblock \href {https://doi.org/10.1063/1.526780}
  {\path{doi:10.1063/1.526780}}.

\bibitem{BATALIN1984106}
Igor~Anatolievich Batalin{~(\cyrillic{Игорь Анатольевич
  Баталин})} and Grigory~Alexandrovich
  Vilkovisky{~(\cyrillic{Григорий Александрович
  Вилковыский})}.
\newblock Closure of the gauge algebra, generalized {L}ie equations and
  {F}eynman rules.
\newblock {\em Nuclear Physics B}, 234(1):106--124, March 1984.
\newblock \href {https://doi.org/10.1016/0550-3213(84)90227-X}
  {\path{doi:10.1016/0550-3213(84)90227-X}}.

\bibitem{Batalin:1981jr}
Igor~Anatolievich Batalin{~(\cyrillic{Игорь Анатольевич
  Баталин})} and Grigory~Alexandrovich
  Vilkovisky{~(\cyrillic{Григорий Александрович
  Вилковыский})}.
\newblock Gauge algebra and quantization.
\newblock {\em Physics Letters B}, 102(1):27--31, June 1981.
\newblock \href {https://doi.org/10.1016/0370-2693(81)90205-7}
  {\path{doi:10.1016/0370-2693(81)90205-7}}.

\bibitem{Batalin:1977pb}
Igor~Anatolievich Batalin{~(\cyrillic{Игорь Анатольевич
  Баталин})} and Grigory~Alexandrovich
  Vilkovisky{~(\cyrillic{Григорий Александрович
  Вилковыский})}.
\newblock Relativistic \({S}\)-matrix of dynamical systems with boson and
  fermion constraints.
\newblock {\em Physics Letters B}, 69(3):309--312, August 1977.
\newblock \href {https://doi.org/10.1016/0370-2693(77)90553-6}
  {\path{doi:10.1016/0370-2693(77)90553-6}}.

\bibitem{Gomes:2023ahz}
Pedro Rogerio~Sergi Gomes.
\newblock An introduction to higher-form symmetries.
\newblock {\em SciPost Physics Lecture Notes}, 74:1, 2023.
\newblock \href {https://arxiv.org/abs/2303.01817} {\path{arXiv:2303.01817}},
  \href {https://doi.org/10.21468/SciPostPhysLectNotes.74}
  {\path{doi:10.21468/SciPostPhysLectNotes.74}}.

\bibitem{Cordova:2022ruw}
Clay Córdova, Thomas~T. Dumitrescu, Kenneth~A. Intriligator, and Shu-Heng
  Shao{~(\begin{CJK*}{UTF8}{bsmi}邵書珩\end{CJK*})}.
\newblock Snowmass white paper: Generalized symmetries in quantum field theory
  and beyond, May 2022.
\newblock \href {https://arxiv.org/abs/2205.09545} {\path{arXiv:2205.09545}},
  \href {https://doi.org/10.48550/arXiv.2205.09545}
  {\path{doi:10.48550/arXiv.2205.09545}}.

\bibitem{Brennan:2023mmt}
Theodore~Daniel Brennan and Sungwoo
  Hong{~(\begin{CJK*}{UTF8}{mj}홍성우\end{CJK*})}.
\newblock Introduction to generalized global symmetries in {QFT} and particle
  physics, June 2023.
\newblock \href {https://arxiv.org/abs/2306.00912} {\path{arXiv:2306.00912}},
  \href {https://doi.org/10.48550/arXiv.2306.00912}
  {\path{doi:10.48550/arXiv.2306.00912}}.

\bibitem{Luo:2023ive}
Ran Luo{~(\begin{CJK*}{UTF8}{gbsn}罗然\end{CJK*})}, Qing-Rui
  Wang{~(\begin{CJK*}{UTF8}{gbsn}王晴睿\end{CJK*})}, and Yi-Nan
  Wang{~(\begin{CJK*}{UTF8}{gbsn}王一男\end{CJK*})}.
\newblock Lecture notes on generalized symmetries and applications.
\newblock {\em Physics Reports}, 1065:1--43, May 2024.
\newblock \href {https://arxiv.org/abs/2307.09215} {\path{arXiv:2307.09215}},
  \href {https://doi.org/10.1016/j.physrep.2024.02.002}
  {\path{doi:10.1016/j.physrep.2024.02.002}}.

\bibitem{Bhardwaj:2023kri}
Lakshya Bhardwaj, Lea~E. Bottini, Ludovic Fraser-Taliente, Liam Gladden, Dewi
  Sid~William Gould, Arthur Platschorre, and Hannah Tillim.
\newblock Lectures on generalized symmetries.
\newblock {\em Physics Reports}, 1051:1--87, February 2024.
\newblock \href {https://arxiv.org/abs/2307.07547} {\path{arXiv:2307.07547}},
  \href {https://doi.org/10.1016/j.physrep.2023.11.002}
  {\path{doi:10.1016/j.physrep.2023.11.002}}.

\bibitem{Eager:2021wpi}
Richard Eager, Fabian Hahner, Ingmar~Akira Saberi, and Brian~R. Williams.
\newblock Perspectives on the pure spinor superfield formalism.
\newblock {\em Journal of Geometry and Physics}, 180:104626, October 2022.
\newblock \href {https://arxiv.org/abs/2111.01162} {\path{arXiv:2111.01162}},
  \href {https://doi.org/10.1016/j.geomphys.2022.104626}
  {\path{doi:10.1016/j.geomphys.2022.104626}}.

\bibitem{dotsenko2019twisting}
Vladimir~Viktorovich Dotsenko{~(\cyrillic{Владимир Викторович
  Доценко})}, Sergey~Viktorovich Shadrin{~(\cyrillic{Сергей
  Викторович Шадрин})}, and Bruno Vallette.
\newblock The twisting procedure, October 2018.
\newblock \href {https://arxiv.org/abs/1810.02941} {\path{arXiv:1810.02941}},
  \href {https://doi.org/10.48550/arXiv.1810.02941}
  {\path{doi:10.48550/arXiv.1810.02941}}.

\bibitem{Dotsenko:2022bbv}
Vladimir~Viktorovich Dotsenko{~(\cyrillic{Владимир Викторович
  Доценко})}, Sergey~Viktorovich Shadrin{~(\cyrillic{Сергей
  Викторович Шадрин})}, and Bruno Vallette.
\newblock {\em {M}aurer--{C}artan Methods in Deformation theory: The Twisting
  Procedure}, volume 488 of {\em London Mathematical Society Lecture Note
  Series}.
\newblock Cambridge University Press, Cambridge, United Kingdom, September
  2023.
\newblock \href {https://arxiv.org/abs/2212.11323} {\path{arXiv:2212.11323}},
  \href {https://doi.org/10.1017/9781108963800}
  {\path{doi:10.1017/9781108963800}}.

\bibitem{Kraft:2022efy}
Andreas Kraft and Jonas Schnitzer.
\newblock An introduction to \({L}_\infty\)-algebras and their homotopy theory
  for the working mathematician.
\newblock {\em Reviews in Mathematical Physics}, 36(01):2330006, February 2024.
\newblock \href {https://arxiv.org/abs/2207.01861} {\path{arXiv:2207.01861}},
  \href {https://doi.org/10.1142/S0129055X23300066}
  {\path{doi:10.1142/S0129055X23300066}}.

\bibitem{Loday:2012aa}
Jean-Louis Loday and Bruno Vallette.
\newblock {\em Algebraic Operads}, volume 346 of {\em
  \foreignlanguage{german}{Grundlehren der mathematischen Wissenschaften}}.
\newblock Springer, Berlin, Germany, August 2012.
\newblock \href {https://doi.org/10.1007/978-3-642-30362-3}
  {\path{doi:10.1007/978-3-642-30362-3}}.

\bibitem{2007arXiv0709.1228G}
Victor~Alexandrovich Ginzburg and Mikhail~Mikhailovich Kapranov.
\newblock Koszul duality for operads.
\newblock {\em Duke Mathematical Journal}, 76(1):203--272, October 1994.
\newblock \href {https://arxiv.org/abs/0709.1228} {\path{arXiv:0709.1228}},
  \href {https://doi.org/10.1215/S0012-7094-94-07608-4}
  {\path{doi:10.1215/S0012-7094-94-07608-4}}.

\bibitem{Dolgushev:2005sp}
Vasiliy~A. Dolgushev{~(\cyrillic{Василий А. Долгушев})}.
\newblock {\em A Proof of {T}sygan's Formality Conjecture for an Arbitrary
  Smooth Manifold}.
\newblock PhD thesis, Massachusetts Institute of Technology, Cambridge,
  Massachusetts, United States of America, June 2005.
\newblock \href {https://arxiv.org/abs/math/0504420}
  {\path{arXiv:math/0504420}}, \href {https://doi.org/1721.1/30354}
  {\path{doi:1721.1/30354}}.

\bibitem{Zwiebach:1992ie}
Barton Zwiebach~Cantor.
\newblock Closed string field theory: Quantum action and the
  {B}atalin-{V}ilkovisky master equation.
\newblock {\em Nuclear Physics B}, 390(1):33--152, January 1993.
\newblock \href {https://arxiv.org/abs/hep-th/9206084}
  {\path{arXiv:hep-th/9206084}}, \href
  {https://doi.org/10.1016/0550-3213(93)90388-6}
  {\path{doi:10.1016/0550-3213(93)90388-6}}.

\bibitem{Markl:1997bj}
Martin Markl.
\newblock Loop homotopy algebras in closed string field theory.
\newblock {\em Communications in Mathematical Physics}, 221(2):367--384, July
  2001.
\newblock \href {https://arxiv.org/abs/hep-th/9711045}
  {\path{arXiv:hep-th/9711045}}, \href {https://doi.org/10.1007/PL00005575}
  {\path{doi:10.1007/PL00005575}}.

\bibitem{Braun:2013lwa}
Christopher Braun and Andrey~Yurievich Lazarev{~(\cyrillic{Андрей
  Юрьевич Лазарев})}.
\newblock {Unimodular homotopy algebras and Chern{\textendash}Simons theory}.
\newblock {\em Journal of Pure and Applied Algebra}, 219(11):5158--5194,
  November 2015.
\newblock \href {https://arxiv.org/abs/1309.3219} {\path{arXiv:1309.3219}},
  \href {https://doi.org/10.1016/j.jpaa.2015.05.017}
  {\path{doi:10.1016/j.jpaa.2015.05.017}}.

\bibitem{Pulmannthesis}
Ján Pulmann.
\newblock S-matrix and homological perturbation lemma.
  \foreignlanguage{czech}{S-matice a homologické perturbační lemma}.
\newblock Master's thesis, \foreignlanguage{czech}{Univerzita Karlova}, Prague,
  Czechia, 2016.
\newblock \href {https://doi.org/20.500.11956/75884}
  {\path{doi:20.500.11956/75884}}.

\bibitem{Doubek:2017naz}
Martin Doubek, Branislav Jurčo, and Ján Pulmann.
\newblock Quantum \({L}_\infty\) algebras and the homological perturbation
  lemma.
\newblock {\em Communications in Mathematical Physics}, 367(1):215--240, April
  2019.
\newblock \href {https://arxiv.org/abs/1712.02696} {\path{arXiv:1712.02696}},
  \href {https://doi.org/10.1007/s00220-019-03375-x}
  {\path{doi:10.1007/s00220-019-03375-x}}.

\bibitem{Doubek:2020rbg}
Martin Doubek, Branislav Jurčo, Martin Markl, and Ivo~Michael Sachs.
\newblock {\em Algebraic Structure of String Field Theory}, volume 973 of {\em
  Lecture Notes in Physics}.
\newblock Springer, Cham, Switzerland, November 2020.
\newblock \href {https://doi.org/10.1007/978-3-030-53056-3}
  {\path{doi:10.1007/978-3-030-53056-3}}.

\bibitem{zbMATH05866383}
Kevin~Joseph Costello.
\newblock {\em Renormalization and Effective Field Theory}, volume 170 of {\em
  Mathematical Surveys and Monographs}.
\newblock American Mathematical Society, Providence, Rhode Island, United
  States of America, 2011.
\newblock \href {https://doi.org/10.1090/surv/170}
  {\path{doi:10.1090/surv/170}}.

\bibitem{AST_1985__S131__257_0}
Jean-Louis Koszul.
\newblock \foreignlanguage{french}{Crochet de {Schouten-Nijenhuis} et
  cohomologie}.
\newblock In {\em \foreignlanguage{french}{Élie Cartan et les mathématiques
  d'aujourd'hui -- Lyon, 25--29 juin 1984}}, number S131 in
  \foreignlanguage{french}{Astérisque}, pages 257--271.
  \foreignlanguage{french}{Société mathématique de France}, Paris, France,
  1985.
\newblock URL: \url{http://www.numdam.org/item/AST_1985__S131__257_0/}.

\bibitem{Akman:1995tm}
Füsun Akman.
\newblock On some generalizations of {B}atalin--{V}ilkovsky algebras.
\newblock {\em Journal of Pure and Applied Algebra}, 120(2):105--141, August
  1997.
\newblock \href {https://arxiv.org/abs/q-alg/9506027}
  {\path{arXiv:q-alg/9506027}}, \href
  {https://doi.org/10.1016/S0022-4049(96)00036-9}
  {\path{doi:10.1016/S0022-4049(96)00036-9}}.

\bibitem{Sati:2008eg}
Hisham Sati, Urs Schreiber, and James~Dillon Stasheff.
\newblock ${L}_\infty$-algebra connections and applications to
  \(\operatorname{String}\)- and {C}hern-{S}imons \(n\)-transport.
\newblock In Bertfried Fauser, Jürgen Tolksdorf, and Eberhard Zeidler,
  editors, {\em Quantum Field Theory: Competitive Models}, pages 303--424.
  Birkhäuser, Basel, Switzerland, December 2009.
\newblock \href {https://arxiv.org/abs/0801.3480} {\path{arXiv:0801.3480}},
  \href {https://doi.org/10.1007/978-3-7643-8736-5_17}
  {\path{doi:10.1007/978-3-7643-8736-5_17}}.

\bibitem{Borsten:2024gox}
Leron Borsten, Mehran Jalali~Farahani, Branislav Jurčo, Hyungrok
  Kim~(\begin{CJK*}{UTF8}{bsmi}金炯錄\end{CJK*}), Jiří Nárožný, Dominik
  Rist, Christian Sämann, and Martin Wolf.
\newblock Higher gauge theory.
\newblock In Richard~Joseph Szabo and Martin Bojowald, editors, {\em
  Encyclopedia of Mathematical Physics. Vol.\ 4}, pages 159--185. Academic
  Press, New York, United States of America, second edition, 2025.
\newblock \href {https://arxiv.org/abs/2401.05275} {\path{arXiv:2401.05275}},
  \href {https://doi.org/10.1016/B978-0-323-95703-8.00217-2}
  {\path{doi:10.1016/B978-0-323-95703-8.00217-2}}.

\bibitem{Batalin:2013xpa}
Igor~Anatolievich Batalin{~(\cyrillic{Игорь Анатольевич
  Баталин})} and Klaus Bering.
\newblock External sources in field--antifield formalism.
\newblock {\em International Journal of Modern Physics A}, 29(10):1450058,
  April 2014.
\newblock \href {https://arxiv.org/abs/1311.3653} {\path{arXiv:1311.3653}},
  \href {https://doi.org/10.1142/S0217751X14500584}
  {\path{doi:10.1142/S0217751X14500584}}.

\bibitem{Gomis:1994he}
Joaquím Gomis~Torné, Jordi París, and Stuart Samuel.
\newblock Antibracket, antifields and gauge theory quantization.
\newblock {\em Physics Reports}, 259(1--2):1--145, August 1995.
\newblock \href {https://arxiv.org/abs/hep-th/9412228}
  {\path{arXiv:hep-th/9412228}}, \href
  {https://doi.org/10.1016/0370-1573(94)00112-G}
  {\path{doi:10.1016/0370-1573(94)00112-G}}.

\bibitem{Peskin:1995ev}
Michael~Edward Peskin and Daniel~V. Schroeder.
\newblock {\em An Introduction to Quantum Field Theory}.
\newblock The Advanced Book Program. Perseus Books, Reading, Massachusetts,
  United States of America, 1995.
\newblock \href {https://doi.org/10.1201/9780429503559}
  {\path{doi:10.1201/9780429503559}}.

\bibitem{Weinberg:1996kr}
Steven Weinberg.
\newblock {\em The Quantum Theory of Fields. Vol.\ 2: Modern Applications}.
\newblock Cambridge University Press, Cambridge, United Kingdom, August 2013.
\newblock \href {https://doi.org/10.1017/CBO9781139644174}
  {\path{doi:10.1017/CBO9781139644174}}.

\bibitem{Chiaffrino:2023wxk}
Christoph Chiaffrino, Talha Ersoy, and Olaf Hohm.
\newblock Holography as homotopy.
\newblock {\em Journal of High Energy Physics}, 2024(09):161, September 2024.
\newblock \href {https://arxiv.org/abs/2307.08094} {\path{arXiv:2307.08094}},
  \href {https://doi.org/10.1007/JHEP09(2024)161}
  {\path{doi:10.1007/JHEP09(2024)161}}.

\bibitem{Alfonsi:2024utl}
Luigi Alfonsi, Leron Borsten, Hyungrok
  Kim~(\begin{CJK*}{UTF8}{bsmi}金炯錄\end{CJK*}), Martin Wolf, and Charles
  A.~S. Young.
\newblock Full {S}-matrices and {W}itten diagrams with (relative)
  \({L}_\infty\)-algebras.
\newblock {\em Journal of High Energy Physics}, 2025, 2025.
\newblock \href {https://arxiv.org/abs/2412.16106} {\path{arXiv:2412.16106}}.

\bibitem{Shifman:1988zk}
Mikhail~Arkadyevich Shifman{~(\cyrillic{Михаил Аркадьевич
  Шифман})}.
\newblock Anomalies in gauge theories.
\newblock {\em Physics Reports}, 209(6):341--378, 1991.
\newblock \href {https://doi.org/10.1016/0370-1573(91)90020-M}
  {\path{doi:10.1016/0370-1573(91)90020-M}}.

\bibitem{Bertlmann:1996xk}
Reinhold~Anton Bertlmann.
\newblock {\em Anomalies in Quantum Field Theory}, volume~91 of {\em
  International Series of Monographs on Physics}.
\newblock Oxford University Press, Oxford, United Kingdom, revised edition,
  November 2000.
\newblock \href {https://doi.org/10.1093/acprof:oso/9780198507628.001.0001}
  {\path{doi:10.1093/acprof:oso/9780198507628.001.0001}}.

\bibitem{Harvey:2005it}
Jeffrey~Alan Harvey.
\newblock {TASI} 2003 lectures on anomalies, September 2005.
\newblock \href {https://arxiv.org/abs/hep-th/0509097}
  {\path{arXiv:hep-th/0509097}}, \href
  {https://doi.org/10.48550/arXiv.hep-th/0509097}
  {\path{doi:10.48550/arXiv.hep-th/0509097}}.

\bibitem{Howe:1990pz}
Paul~S. Howe, Ulf~Gustaf Lindström, and Peter~L. White.
\newblock Anomalies and renormalisation in the {BRST--BV} framework.
\newblock {\em Physics Letters B}, 246(3--4):430--434, August 1990.
\newblock \href {https://doi.org/10.1016/0370-2693(90)90625-G}
  {\path{doi:10.1016/0370-2693(90)90625-G}}.

\bibitem{Saberi:2021weg}
Ingmar~Akira Saberi and Brian~R. Williams.
\newblock Twisting pure spinor superfields, with applications to supergravity.
\newblock {\em Pure and Applied Mathematics Quarterly}, 20(2):645--701, 2024.
\newblock \href {https://arxiv.org/abs/2106.15639} {\path{arXiv:2106.15639}},
  \href {https://doi.org/10.4310/PAMQ.2024.v20.n2.a2}
  {\path{doi:10.4310/PAMQ.2024.v20.n2.a2}}.

\bibitem{Kajiura:2004xu}
Hiroshige
  Kajiura~{(\begin{CJK*}{UTF8}{mj}梶\end{CJK*}\begin{CJK*}{UTF8}{bsmi}浦宏成\end{CJK*})}
  and James~Dillon Stasheff.
\newblock Homotopy algebras inspired by classical open-closed string field
  theory.
\newblock {\em Communications in Mathematical Physics}, 263(3):553--581, May
  2006.
\newblock \href {https://arxiv.org/abs/math/0410291}
  {\path{arXiv:math/0410291}}, \href
  {https://doi.org/10.1007/s00220-006-1539-2}
  {\path{doi:10.1007/s00220-006-1539-2}}.

\bibitem{Kajiura:2005sn}
Hiroshige
  Kajiura~{(\begin{CJK*}{UTF8}{mj}梶\end{CJK*}\begin{CJK*}{UTF8}{bsmi}浦宏成\end{CJK*})}
  and James~Dillon Stasheff.
\newblock Open-closed homotopy algebra in mathematical physics.
\newblock {\em Journal of Mathematical Physics}, 47(2):023506, February 2006.
\newblock \href {https://arxiv.org/abs/hep-th/0510118}
  {\path{arXiv:hep-th/0510118}}, \href {https://doi.org/10.1063/1.2171524}
  {\path{doi:10.1063/1.2171524}}.

\bibitem{Kajiura:2006mt}
Hiroshige
  Kajiura~{(\begin{CJK*}{UTF8}{mj}梶\end{CJK*}\begin{CJK*}{UTF8}{bsmi}浦宏成\end{CJK*})}
  and James~Dillon Stasheff.
\newblock {Homotopy algebra of open-closed strings}.
\newblock {\em Geometry \& Topology Monographs}, 13:229--259, 2008.
\newblock \href {https://arxiv.org/abs/hep-th/0606283}
  {\path{arXiv:hep-th/0606283}}, \href
  {https://doi.org/10.2140/gtm.2008.13.229}
  {\path{doi:10.2140/gtm.2008.13.229}}.

\bibitem{PMIHES_1977__47__269_0}
Dennis~Parnell Sullivan.
\newblock Infinitesimal computations in topology.
\newblock {\em \foreignlanguage{french}{Publications mathématiques de
  l'Institut des hautes études scientifiques}}, 47(1):269--331, December 1977.
\newblock \href {https://doi.org/10.1007/BF02684341}
  {\path{doi:10.1007/BF02684341}}.

\bibitem{TolleyThesis}
Melissa~Marie Tolley.
\newblock {\em The Connections between \({A}_\infty\) and \({L}_\infty\)
  Algebras}.
\newblock PhD thesis, North Carolina State University, Raleigh, North Carolina,
  United States of America, March 2013.
\newblock URL: \url{https://www.lib.ncsu.edu/resolver/1840.16/8434}.

\bibitem{LadaTolley}
Thomas~Joseph Lada and Melissa~Marie Tolley.
\newblock Derivations of homotopy algebras.
\newblock {\em \foreignlanguage{latin}{Archivum Mathematicum}}, 49(5):309--315,
  2013.
\newblock \href {https://doi.org/10.5817/AM2013-5-309}
  {\path{doi:10.5817/AM2013-5-309}}.

\bibitem{2014arXiv1409.1691D}
Martin Doubek and Thomas~Joseph Lada.
\newblock Homotopy derivations.
\newblock {\em Journal of Homotopy and Related Structures}, 11(3):599--630,
  September 2016.
\newblock \href {https://arxiv.org/abs/1409.1691} {\path{arXiv:1409.1691}},
  \href {https://doi.org/10.1007/s40062-015-0118-7}
  {\path{doi:10.1007/s40062-015-0118-7}}.

\bibitem{Jonsson:2024rie}
David Simon~Henrik Jonsson.
\newblock ({T}wisted) canonical supermultiplets and their resolutions as
  open-closed homotopy algebras, August 2024.
\newblock \href {https://arxiv.org/abs/2408.15102} {\path{arXiv:2408.15102}},
  \href {https://doi.org/10.48550/arXiv.2408.15102}
  {\path{doi:10.48550/arXiv.2408.15102}}.

\bibitem{Lindstrom:2007rq}
Ulf~Gustaf Lindström.
\newblock Hyperkähler metrics from projective superspace, March 2007.
\newblock \href {https://arxiv.org/abs/hep-th/0703181}
  {\path{arXiv:hep-th/0703181}}, \href
  {https://doi.org/10.48550/arXiv.hep-th/0703181}
  {\path{doi:10.48550/arXiv.hep-th/0703181}}.

\bibitem{Kuzenko:2010bd}
Sergei~M. Kuzenko{~(\cyrillic{Сергей М. Кузенко})}.
\newblock Lectures on nonlinear sigma-models in projective superspace.
\newblock {\em Journal of Physics A}, 43(44):443001, November 2010.
\newblock \href {https://arxiv.org/abs/1004.0880} {\path{arXiv:1004.0880}},
  \href {https://doi.org/10.1088/1751-8113/43/44/443001}
  {\path{doi:10.1088/1751-8113/43/44/443001}}.

\bibitem{Howe:1995md}
Paul~S. Howe and G.~G. Hartwell.
\newblock A superspace survey.
\newblock {\em Classical and Quantum Gravity}, 12(8):1823--1880, August 1995.
\newblock \href {https://doi.org/10.1088/0264-9381/12/8/005}
  {\path{doi:10.1088/0264-9381/12/8/005}}.

\bibitem{Galperin:2001seg}
Alexander~S. Galperin{~(\cyrillic{Александр С. Гальперин})},
  Evgeny~Alexeyevich Ivanov{~(\cyrillic{Евгений Алексеевич
  Иванож})}, Victor~Isaakovich Ogievetsky{~(\cyrillic{Виктор
  Исаакович Огиевецкий})}, and Emeri~S.
  Sokatchev{~(\cyrillic{Емери С. Сокачев})}.
\newblock {\em Harmonic Superspace}.
\newblock Cambridge Monographs on Mathematical Physics. Cambridge University
  Press, Cambridge, United Kingdom, February 2007.
\newblock \href {https://doi.org/10.1017/CBO9780511535109}
  {\path{doi:10.1017/CBO9780511535109}}.

\bibitem{Berkovits:2017ldz}
Nathan Berkovits and Humberto Gómez~Zuñiga.
\newblock An introduction to pure spinor superstring theory.
\newblock In Alexander Cardona~Guio, Pedro~Fernando Morales-Almazán, Hernán
  Ocampo~Duran, Sylvie Paycha, and Andrés~Fernando Reyes~Lega, editors, {\em
  9th Summer School on Geometric, Algebraic and Topological Methods for Quantum
  Field Theory}, Mathematical Physics Studies, pages 221--246, Cham,
  Switzerland, 2017. Springer.
\newblock \href {https://arxiv.org/abs/1711.09966} {\path{arXiv:1711.09966}},
  \href {https://doi.org/10.1007/978-3-319-65427-0_6}
  {\path{doi:10.1007/978-3-319-65427-0_6}}.

\bibitem{Cederwall:2013vba}
Nils Martin~Sten Cederwall.
\newblock Pure spinor superfields -- an overview.
\newblock In Stefano Bellucci, editor, {\em Breaking of Supersymmetry and
  Ultraviolet Divergences in Extended Supergravity. Proceedings of the
  INFN-Laboratori Nazionali di Frascati School 2013}, volume 153 of {\em
  Springer Proceedings in Physics}, pages 61--93, Cham, Switzerland, 2014.
  Springer.
\newblock \href {https://arxiv.org/abs/1307.1762} {\path{arXiv:1307.1762}},
  \href {https://doi.org/10.1007/978-3-319-03774-5_4}
  {\path{doi:10.1007/978-3-319-03774-5_4}}.

\bibitem{Cederwall:2022fwu}
Nils Martin~Sten Cederwall.
\newblock Pure spinors in classical and quantum supergravity.
\newblock In Cosimo Bambi, Leonardo Modesto, and Ilya~Lvovich
  Shapiro{~(\cyrillic{Илья Львович Шапиро})}, editors, {\em
  Handbook of Quantum Gravity}. Springer, Singapore, 2023.
\newblock \href {https://arxiv.org/abs/2210.06141} {\path{arXiv:2210.06141}},
  \href {https://doi.org/10.1007/978-981-19-3079-9_47-1}
  {\path{doi:10.1007/978-981-19-3079-9_47-1}}.

\bibitem{Jonsson:2024uyr}
David Simon~Henrik Jonsson, Hyungrok
  Kim~(\begin{CJK*}{UTF8}{bsmi}金炯錄\end{CJK*}), and Charles
  Alastair~Stephen Young.
\newblock Homotopy representations of extended holomorphic symmetry in
  holomorphic twists, August 2024.
\newblock \href {https://arxiv.org/abs/2408.00704} {\path{arXiv:2408.00704}}.

\bibitem{Baulieu:2010ch}
Laurent Baulieu.
\newblock \(\operatorname{SU}(5)\)-invariant decomposition of ten-dimensional
  {Y}ang--{M}ills supersymmetry.
\newblock {\em Physics Letters B}, 698(1):63--67, March 2011.
\newblock \href {https://arxiv.org/abs/1009.3893} {\path{arXiv:1009.3893}},
  \href {https://doi.org/10.1016/j.physletb.2010.12.044}
  {\path{doi:10.1016/j.physletb.2010.12.044}}.

\bibitem{Losev:2023gsq}
Andrey~Semenovich Losev{~(\cyrillic{Андрей Семенович
  Лосев})} and Vyacheslav Lysov{~(\cyrillic{Вячеслав
  Лысов})}.
\newblock {BV}-refinement of the on-shell supersymmetry and localization,
  December 2023.
\newblock \href {https://arxiv.org/abs/2312.13999} {\path{arXiv:2312.13999}},
  \href {https://doi.org//10.48550/arXiv.2312.13999}
  {\path{doi:/10.48550/arXiv.2312.13999}}.

\bibitem{Cattaneo:2025wdw}
Alberto~Sergio Cattaneo and Shuhan Jiang.
\newblock Equivariant localization in {B}atalin-{V}ilkovisky formalism, January
  2025.
\newblock \href {https://arxiv.org/abs/2501.17082} {\path{arXiv:2501.17082}},
  \href {https://doi.org/10.48550/arXiv.2501.17082}
  {\path{doi:10.48550/arXiv.2501.17082}}.

\bibitem{localisationpaper}
Alexandros~Spyridion Arvanitakis{~({\textgreekfont Αλέξανδρος
  Σπυριδίων Αρβανιτάκης})}, Leron Borsten, Dimitri
  Kanakaris~Decavel, and Hyungrok~Kim
  (\begin{CJK*}{UTF8}{bsmi}金炯錄\end{CJK*}).
\newblock {BV} enhanced localisation in supersymmetric quantum field theory.
\newblock To appear, 2025.

\bibitem{Ferrara:1978em}
Sergio Ferrara and Peter van Nieuwenhuizen.
\newblock The auxiliary fields of supergravity.
\newblock {\em Physics Letters B}, 74(4--5):333--335, April 1978.
\newblock \href {https://doi.org/10.1016/0370-2693(78)90670-6}
  {\path{doi:10.1016/0370-2693(78)90670-6}}.

\bibitem{Stelle:1978ye}
Kellogg~S. Stelle and Peter~Christopher West.
\newblock Minimal auxiliary fields for supergravity.
\newblock {\em Physics Letters B}, 74(4--5):330--332, April 1978.
\newblock \href {https://doi.org/10.1016/0370-2693(78)90669-X}
  {\path{doi:10.1016/0370-2693(78)90669-X}}.

\bibitem{Fradkin:1978jq}
Efim~Samoilovich Fradkin{~(\cyrillic{Ефим Самойлович
  Фрадкин})} and Mikhail~Andreevich Vasiliev{~(\cyrillic{Михаил
  Андреевич Васильев})}.
\newblock \({S}\)-matrix for theories that admit closure of the algebra with
  the aid of auxiliary fields. {A}uxiliary fields in supergravity.
\newblock {\em \foreignlanguage{italian}{Lettere al Nuovo Cimento}},
  22(16):651--659, August 1978.
\newblock \href {https://doi.org/10.1007/BF02783437}
  {\path{doi:10.1007/BF02783437}}.

\bibitem{Sohnius:1981tp}
Martin~F. Sohnius and Peter~Christopher West.
\newblock An alternative minimal off-shell version of \({N}=1\) supergravity.
\newblock {\em Physics Letters B}, 105(5):353--357, October 1981.
\newblock \href {https://doi.org/10.1016/0370-2693(81)90778-4}
  {\path{doi:10.1016/0370-2693(81)90778-4}}.

\bibitem{Castellani:1991eu}
Leonardo Castellani, Riccardo D'Auria, and Pietro~Giuseppe Frè.
\newblock {\em Supergravity and Superstrings: A Geometric Perspective}.
\newblock World Scientific, Singapore, 1991.
\newblock \href {https://doi.org/10.1142/0224} {\path{doi:10.1142/0224}}.

\bibitem{Freedman:2012zz}
Daniel~Zissel Freedman and Antoine Van~Proeyen.
\newblock {\em Supergravity}.
\newblock Cambridge University Press, Cambridge, United Kingdom, April 2012.
\newblock \href {https://doi.org/10.1017/CBO9781139026833}
  {\path{doi:10.1017/CBO9781139026833}}.

\bibitem{Tanii:2014gaa}
Yoshiaki Tanii~{(\begin{CJK*}{UTF8}{bsmi}谷井義彰\end{CJK*})}.
\newblock {\em Introduction to Supergravity}, volume~1 of {\em Springer Briefs
  in Mathematical Physics}.
\newblock Springer, Tokyo, Japan, August 2014.
\newblock \href {https://doi.org/10.1007/978-4-431-54828-7}
  {\path{doi:10.1007/978-4-431-54828-7}}.

\bibitem{Ortin:2015hya}
Tomas Ortín~Miguel.
\newblock {\em Gravity and Strings}.
\newblock Cambridge Monographs on Mathematical Physics. Cambridge University
  Press, Cambridge, United Kingdom, second edition, March 2015.
\newblock \href {https://doi.org/10.1017/CBO9781139019750}
  {\path{doi:10.1017/CBO9781139019750}}.

\end{thebibliography}
\end{document}